\documentclass[a4paper,fleqn]{cas-dc}
\usepackage{amsmath,amssymb,amsfonts}
\usepackage{amsthm}
\usepackage{nicefrac}
\usepackage{algorithmic}
\usepackage{textcomp}
\usepackage{xcolor}
\usepackage[utf8]{inputenc} 
\usepackage[T1]{fontenc}
\usepackage{graphicx}
\usepackage[font=small]{caption}
\usepackage{tikz}
\usetikzlibrary{positioning}
\usepackage{pgfplots}
\usepackage{epstopdf}
\usepackage{placeins}
\usepackage{float}
\definecolor{darkgreen}{rgb}{0.0, 0.5, 0.0}
\usepackage[labelformat=simple]{subcaption}
\usepackage{multirow}
\usepackage{lipsum}
\usepackage{mathtools}
\usepackage{dsfont}
\usepackage{enumitem}

\def\BibTeX{{\rm B\kern-.05em{\sc i\kern-.025em b}\kern-.08em
    T\kern-.1667em\lower.7ex\hbox{E}\kern-.125emX}}

\input{mysymbol.sty}

\newtheorem{lemma}{Lemma}
\newtheorem{proposition}{Proposition}

\newtheorem{definition}{Definition}
\newtheorem{remark}{Remark}

\newtheorem{Problem}{Problem}

\pgfplotsset{compat=1.18}

\usepackage[numbers,sort&compress]{natbib}
\begin{document}
\let\WriteBookmarks\relax
\def\floatpagepagefraction{1}
\def\textpagefraction{.001}

\shorttitle{Ant Backpressure Routing for Dynamic Wireless Multi-hop Networks}
\shortauthors{Erfaniantaghvayi et al.}

\title[mode=title]{Ant Backpressure Routing for Dynamic Wireless Multi-hop Networks with Mixed Traffic Patterns}

\author[1]{Negar Erfaniantaghvayi}[orcid=0000-0002-8049-0067]
\ead{ne12@rice.edu}
\credit{Methodology, software, investigation, validation, writing -- original draft}

\author[1]{Zhongyuan Zhao}[orcid=0000-0003-0346-8015]
\cormark[1]
\ead{zhongyuan.zhao@rice.edu}
\credit{Conceptualization, methodology, software, formal analysis, supervision, writing -- original draft, writing -- review \& editing}

\author[2]{Kevin Chan}[orcid=0000-0002-6425-5403]
\ead{kevin.s.chan.civ@army.mil}
\credit{Conceptualization, validation, writing -- review \& editing}

\author[2]{Ananthram Swami}[orcid=0000-0003-1439-332X]
\ead{ananthram.swami.civ@army.mil}
\credit{Conceptualization, resources, writing -- review \& editing}

\author[1]{Santiago Segarra}[orcid=0000-0002-8408-9633]
\cormark[1]
\ead{segarra@rice.edu}
\credit{Supervision, project administration, funding acquisition, writing -- review \& editing}

\affiliation[1]{organization={Department of Electrical and Computer Engineering, Rice University},
                addressline={6100 Main Street}, 
                city={Houston},
                postcode={77005}, 
                state={TX},
                country={USA}}

\affiliation[2]{organization={DEVCOM Army Research Laboratory},
                addressline={2800 Powder Mill Road}, 
                city={Adelphi},
                postcode={20783}, 
                state={MD},
                country={USA}}

% \cortext[cor1]{Corresponding author at: Department of Electrical and Computer Engineering, Rice University, United States.}
\cortext[cor1]{Corresponding authors.}
% \cortext[cor2]{Corresponding author}

% \tnotemark[1]
% \tnotetext[1]{Research was sponsored by the DEVCOM ARL Army Research Office and was accomplished under Cooperative Agreement Number W911NF-24-2-0008 and W911NF-19-2-0269. The views and conclusions contained in this document are those of the authors and should not be interpreted as representing the official policies, either expressed or implied, of the Army Research Office or the U.S. Government. The U.S. Government is authorized to reproduce and distribute reprints for Government purposes notwithstanding any copyright notation herein.}

\fnmark[1]
\fntext[1]{Preliminary results were presented in IEEE MILCOM 2024~\cite{Erfaniantaghvayi2024}.}

% Option 1: paste the abstract here.
% Option 2: replace this block with 
\begin{abstract}
Backpressure (BP) routing and its shortest-path biased variant (SP-BP) provide powerful congestion-aware multipath resource allocation for wireless multi-hop networks, 
but they rely on per-commodity queueing and slot-by-slot control that may be difficult to realize under practical or legacy forwarding architectures. 
Moreover, even state-of-the-art SP-BP still suffers from the last-packet problem when short-lived traffic coexists with streaming flows.
To address these limitations, we propose Ant Backpressure (Ant-BP), a periodic and fully distributed routing scheme that decouples route learning from packet forwarding. 
Ant-BP uses virtual SP-BP to construct pheromone-based forwarding probabilities, while actual packets are forwarded through per-neighbor first-in-first-out (FIFO) queues with probabilistic next-hop selection.  
This architecture enables link-capacity sharing across commodities, mitigates starvation of short-lived traffic, and extends the benefits of SP-BP to network architectures based on per-neighbor FIFO forwarding. 
Through periodic virtual updates, Ant-BP also adapts to transient link failures and mobility-induced topology changes.
Our theoretical analysis and simulations show that, compared with conventional ant colony optimization (ACO) routing, virtual SP-BP enables Ant-BP to establish higher-quality forwarding policies with lower overhead.
As a result, Ant-BP improves latency and delivery ratio over SP-BP and ACO-based baselines under mixed streaming and bursty traffic, achieves throughput comparable to SP-BP at low and medium traffic load, and remains robust to mismatched virtual-traffic assumptions, transient link failures, and node mobility.
\end{abstract}

% Uncomment if required by the target Elsevier journal.
% \begin{highlights}
% \item Insert highlight 1.
% \item Insert highlight 2.
% \item Insert highlight 3.
% \end{highlights}


\maketitle

% Main manuscript body: sections only, no \documentclass, no \begin{document}, no bibliography style.

\section{Introduction}
Wireless multi-hop networks have transcended their traditional applications in military communications, disaster relief, and wireless sensor networks, emerging as an enabler of next-generation (xG) networks. 
Infrastructure-light architectures are now essential to emerging paradigms in highly dynamic environments, such as integrated access and backhaul (IAB), non-terrestrial networks, device-to-device and massive machine-type communications, connected vehicles, and robotic swarms~\cite{lin2006tutorial,sarkar2013ad,Patriciello2016,kott2016internet,Cudak2021,akyildiz20206g}.
To ensure scalability and adaptivity, these applications increasingly rely on self-organizing, distributed solutions for resource allocation and network orchestration.
Backpressure (BP) routing and its variants~\cite{tassiulas1990stability,ryu2010back,georgiadis2006resource,neely2005dynamic,zhao2023icassp, zhao2023enhanced, zhao2024tmlcn,jiao2015virtual, moeller2010routing,Alresaini2016bp,ji2012delay,liaskos2023analysis} have emerged as promising solutions in this domain. 
In BP schemes, packets destined for a specific node constitute a commodity, requiring per-commodity queues at each node. 
In contrast to traditional decoupled protocols, BP makes joint routing and scheduling decisions through Max-Weight scheduling driven entirely by local queue differentials (pressure)~\cite{tassiulas1990stability}. 
By activating non-conflicting links that maximize instantaneous network pressure, BP dynamically forwards traffic along congestion gradients, effectively exploring all possible routes. 
Rooted in Lyapunov drift theory, this approach mathematically guarantees maximum queue stability~\cite{tassiulas1990stability,neely2005dynamic,georgiadis2006resource,zhao2024tmlcn}, providing the critical adaptivity needed to absorb demand shocks. 
Furthermore, with fully distributed greedy schedulers~\cite{joo2011local,zhao2022link}, BP schemes enable self-organizing architectures solely relying on local interactions among neighboring nodes.

Despite its theoretical elegance, the real-world deployment of classical BP routing is hindered by its high latency and cross-layer requirements. 
BP suffers latency degradation under low-to-medium traffic loads due to well-documented drawbacks: slow startup, random walks, and the last-packet problem~\cite{jiao2015virtual,Alresaini2016bp,ji2012delay,moeller2010routing}. 
While recent advances in shortest-path biased BP (SP-BP) routing have successfully mitigated the slow startup and random walk issues by shaping routes via distance bias fields~\cite{neely2005dynamic, georgiadis2006resource,zhao2023icassp,zhao2023enhanced,zhao2024tmlcn}, three structural challenges remain. 
First, short-lived traffic that lacks a persistent congestion gradient can be starved in the presence of heavy streaming flows (the last-packet problem).
Second, the exclusive commodity selection in BP under-utilizes link capacity under low-to-medium traffic~\cite{jiao2015virtual,Alresaini2016bp,ji2012delay,moeller2010routing,zhao2026ton}.
More critically, its cross-layer operation and memory-intensive per-commodity queueing system conflict with legacy network protocols with decoupled forwarding and route optimization, creating massive architectural friction~\cite{moeller2010routing,bui2011novel,Liaskos2023tnse}.

Real-world networks predominantly rely on standard forwarding architectures originally designed for stable wired infrastructure. 
Legacy Internet protocols (IP)~\cite{moy1998ospf} and Software-Defined Networks (SDNs)~\cite{mckeown2008openflow,kreutz2014software} utilize simplified, per-neighbor first-in-first-out (FIFO) queues driven by shortest-path routing or centralized traffic engineering~\cite{oliveira2011mathematical}. 
While perfectly suited for the reliable links of wired domains, their direct adoption into dynamic wireless multi-hop networks exposes severe limitations. 
Traditional IP and SDN assume static network topology and fast control plane, which break down quickly in mobile ad hoc networks (MANETs)~\cite{abolhasan2004review,chlamtac2003mobile,haque2016wireless,vinayakray2012routing, aggarwal2014performance}. 
Alternatively, bio-inspired distributed algorithms like Ant Colony Optimization (ACO)~\cite{di2005anthocnet, Purkayastha2013convergence,zhang2017survey,dorigo2019ant} enable probabilistic multipath routing, but suffer from prohibitively high control overhead and slow convergence during link failures or node mobility~\cite{vinayakray2012routing, aggarwal2014performance}.

To tackle these challenges, we propose Ant Backpressure (Ant-BP) routing, a hybrid approach that incorporates the effectiveness of SP-BP in gradient-based route discovery into standard forwarding architectures~\cite{Erfaniantaghvayi2024}. 
A lightweight SP-BP in the virtual plane periodically maintains pheromone policies for probabilistic forwarding, by exchanging only packet counts--without requiring feedback trips of scout ants in conventional ACO~\cite{di2005anthocnet, Purkayastha2013convergence,zhang2017survey,dorigo2019ant}. 
Real data packets are then forwarded through standard per-neighbor FIFO queues, enabling link-capacity sharing across commodities. 
Furthermore, by mapping bursty flows in data plane into streaming virtual traffic, Ant-BP effectively mitigates the last-packet problem in both virtual and data planes. 
Compared to conventional ACO, the gradient-based route discovery of Ant-BP incurs significantly lower overhead, making it adaptive to transient link failures and node mobility.

\vspace{1mm}
\noindent
\textbf{Contribution.} The contributions of this paper are fourfold:
\begin{enumerate}[leftmargin=1em,labelsep=0.2em,itemsep=0pt,topsep=1pt]
\item \textit{Decoupled hybrid architecture}: We propose Ant-BP, a distributed routing scheme that combines the gradient-driven route learning of SP-BP with lightweight, probabilistic packet forwarding. 
By using virtual SP-BP to construct pheromone policies and per-neighbor FIFO queues for physical forwarding, Ant-BP mitigates short-flow starvation and extends SP-BP-style intelligence to legacy-compatible network architectures.
\item \textit{Theoretical interpretation}: We theoretically ground Ant-BP by relating SP-BP to a stochastic cost-minimization problem in the fully dynamic regime. 
We demonstrate that Ant-BP heuristically approximates the corresponding long-term routing objective under a restricted FIFO forwarding regime, trading a reduced strict stability region for practical implementability.
\item \textit{Adaptive mechanisms for network dynamics}: We design mechanisms to make Ant-BP robust in dynamic wireless environments. 
These include periodic virtual route re-estimation, pheromone penalization to handle transient link failures, and mobility adaptation through remapping of stranded physical packets into the virtual state.
\item \textit{Comprehensive numerical evaluation}: Through extensive simulations, we show that Ant-BP improves end-to-end latency and delivery ratios over SP-BP and ACO-based baselines under mixed streaming and bursty traffic. 
Furthermore, it achieves throughput comparable to SP-BP at low-to-medium loads and remains highly robust to mismatched virtual-traffic assumptions, link failures, and mobility-induced topology changes.
\end{enumerate}

\vspace{1mm}
\noindent
{\bf Notation.} 
%The notational convention is as follows:
Operators 
$ (\cdot)^\top $, $ \odot $, and $ |\cdot| $ represent the transpose operator, Hadamard (element-wise) product operator, and the cardinality of a set.
$ \mathds{1}(\cdot) $ is the indicator function. Upright bold lower-case symbol, e.g., $\bbz$, denotes a column vector, and $\bbz_i$ denotes the $i$-th element of vector $\bbz$. 
Upright bold upper-case symbol $\bbZ$ denotes a matrix, whose element at row $i$ and column $j$ is denoted by $\bbZ_{ij}$, the entire row $i$ by $\bbZ_{i*}$, and the entire column $j$ by $\bbZ_{*j}$.
A capital calligraphic letter denotes a set, queue, or graph.
$ \ccalN_{\ccalG}(i) $ represents the set of immediate neighbors of node $i$ on graph $\ccalG$.

\begin{figure*}[!t]
    \includegraphics[width=0.95\linewidth]{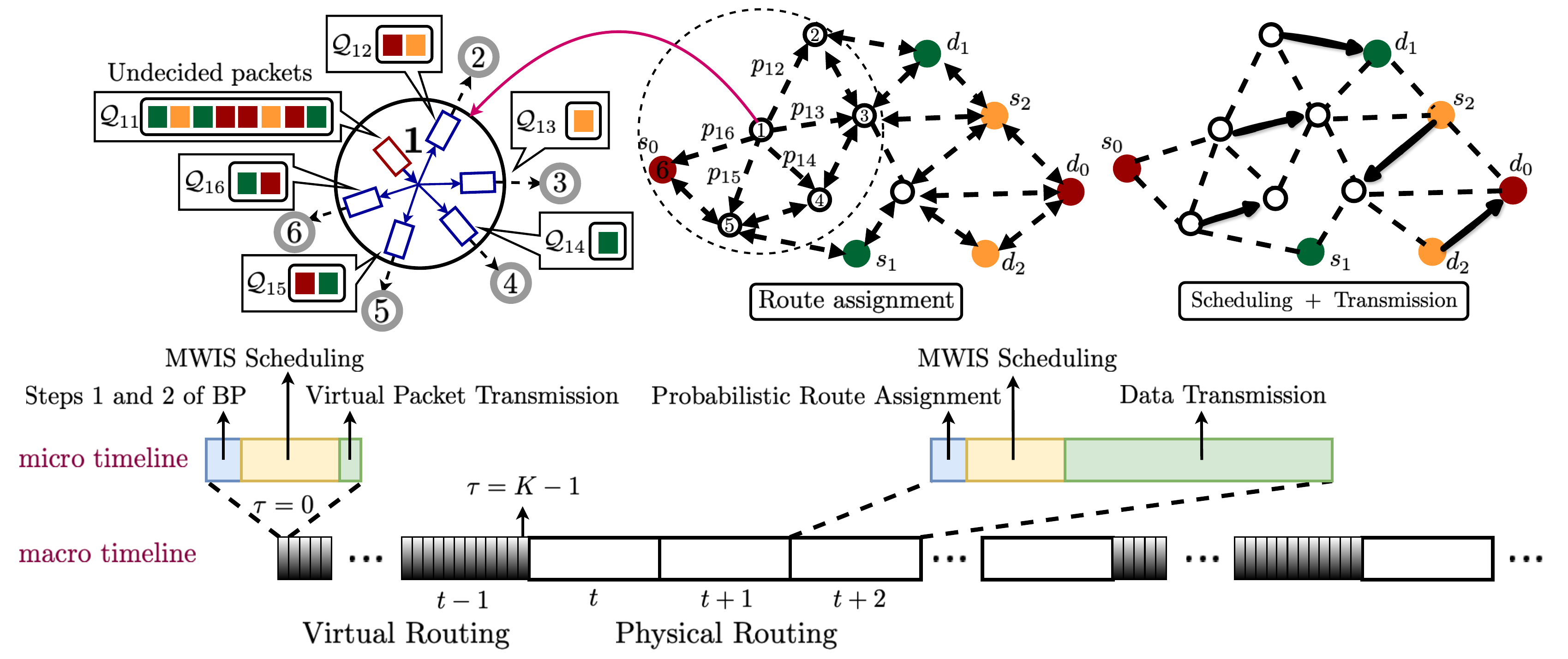}
    \vspace{-0.1in}
    \caption{Ant backpressure routing system diagram: queueing system design, operations, and timelines.
    }
    \label{fig:timeline}
\end{figure*}

\section{Related Work}\label{sec:survey}
Real-world wireless multi-hop networks often struggle to adopt standard forwarding architectures—such as legacy IP or Software-Defined Networks (SDNs)—because their reliance on static topologies and centralized control breaks down under intermittent connectivity. 
To address this, early work focused on topology-based routing protocols designed specifically for mobile ad hoc networks (MANETs)~\cite{abolhasan2004review, chlamtac2003mobile}. 
Representative examples include Dynamic Source Routing (DSR)~\cite{johnson1996dynamic, johnson2001dsr, leung2001mp} and Ad hoc On-demand Distance Vector (AODV)~\cite{lee2003scalability, saravanan2019modeling, marina2001demand, perkins1999aodv}, which establish routes through on-demand route discovery and maintain them using control messages. 
While such protocols are simple and scalable, they rely on predetermined, rigid paths that suffer from severe performance degradation in dynamic environments where link rates and congestion levels change rapidly~\cite{vinayakray2012routing, aggarwal2014performance}.

To improve robustness against such channel variability, opportunistic routing techniques were later proposed to dynamically select forwarding nodes based on link quality, congestion conditions, and priority~\cite{bhorkar2011adaptive, zhang2013overview,chakchouk2015survey}. 
Protocols such as ExOR~\cite{biswas2005exor} and CORMAN~\cite{wang2012corman} exploit the broadcast nature of the wireless medium to dynamically select next-hop relays, while subsequent congestion-aware variants like D-ORCD~\cite{bhorkar2015opportunistic, kathiravelu2011enhanced, shelke2017congestion} select forwarding nodes based on expected draining times. 
However, while these heuristic approaches increase transmission reliability, they lack the rigorous mathematical framework required to guarantee maximum network-wide queue stability under heavy congestion.

Another important class of routing strategies relies on bio-inspired swarm intelligence, particularly ACO \cite{di2005anthocnet, Purkayastha2013convergence, zhang2017survey, dorigo2019ant}, which discovers and maintains multiple paths in a distributed, probabilistic manner. 
Protocols like AntHocNet \cite{di2005anthocnet} combine reactive path discovery with proactive route maintenance, where artificial ants explore the network and update pheromone routing tables \cite{dorigo2019ant}. 
While various extensions attempt to improve path stability, delay performance, and energy efficiency in dynamic wireless environments~\cite{zhang2017survey, sharma2017survey}, conventional ACO algorithms face severe structural limitations. 
Specifically, because route exploration and maintenance require the continuous transmission of payload-heavy physical scout ants that must traverse the network and return to their sources, these protocols suffer from prohibitively high control overhead and slow convergence~\cite{Purkayastha2013convergence} in MANETs with transient link failures or node mobility. 

% Another important class of routing strategies in wireless networks relies on bio-inspired or swarm-intelligence-based algorithms. Among these, ACO~\cite{di2005anthocnet, zhang2017survey,dorigo2019ant} has received significant attention due to its ability to discover and maintain multiple paths in a distributed manner. 
% In ACO-based routing, artificial ants explore the network and update probabilistic routing tables through pheromone values that reflect the quality of discovered routes \cite{dorigo2019ant}. 
% AntHocNet is one of the most well-known ACO-based routing protocols for wireless ad hoc networks, combining reactive path discovery with proactive route maintenance to adapt to topology changes \cite{di2005anthocnet}. 
% A number of extensions and variants have been proposed to improve path stability, delay performance, and energy efficiency in dynamic wireless environments \cite{zhang2017survey, sharma2017survey}. 
% However, ACO-based routing protocols \update{mostly focus on practical challenges in reducing the overhead of route exploration and maintenance, while the underlying ACO algorithms suffer slow convergence in wireless networks}.
% typically rely on additional control packets for route exploration and maintenance, which can introduce non-negligible signaling overhead.

In parallel with these approaches, queue-based routing frameworks~\cite{rehman2016joint, naseem2017queue}--most notably BP routing~\cite{tassiulas1990stability,ryu2010back,georgiadis2006resource,neely2005dynamic,zhao2023icassp, zhao2023enhanced, zhao2024tmlcn,jiao2015virtual, moeller2010routing,Alresaini2016bp,ji2012delay,liaskos2023analysis}--emerged to jointly address routing and scheduling under wireless interference constraints. 
Due to the theoretical stability guarantees of BP, extensive research has focused on mitigating its high latency~\cite{ryu2010back,georgiadis2006resource,neely2005dynamic,zhao2023icassp, zhao2023enhanced, zhao2024tmlcn,jiao2015virtual, moeller2010routing,Alresaini2016bp,ji2012delay,liaskos2023analysis}. 
Proposed solutions include incorporating delay-based backlogs into pressure calculation and scheduling~\cite{maglaras2014delay,jiao2015virtual, hai2017delay}, using LIFO queues~\cite{moeller2010routing}, employing multi-stage scheduling with shortest-path metrics~\cite{shan2022tsbs}, and introducing (delay-aware) shortest-path bias fields to shape routing behaviors~\cite{georgiadis2006resource,neely2005dynamic,zhao2023icassp, zhao2023enhanced, zhao2024tmlcn}. 
Furthermore, to address the architectural friction of per-commodity queues, virtual or shadow queue architectures were proposed to mathematically decouple optimal route calculation from physical packet queues~\cite{bui2011novel,Liaskos2023tnse}.

Despite these advances, fundamental structural issues prevent real-world adoption. 
First, the last-packet problem persists due to per-commodity queue design~\cite{jiao2015virtual,Alresaini2016bp,ji2012delay,moeller2010routing,Erfaniantaghvayi2024} which becomes more significant as short-lived traffic dominates machine-type communications. 
More critically, the requirement of cross-layer control significantly hinders the deployment of BP schemes in networks relying on standard low-level network protocols (e.g., Ethernet) based on simple per-neighbor FIFO forwarding~\cite{Kawadia2005caution}.
To close the gap between these disjointed paradigms, we propose to combine the lightweight multipath route establishment of SP-BP with the FIFO-based probabilistic packet forwarding for dynamic wireless multi-hop networks.

\section{System Model}\label{sec:system}

\subsection{Connectivity and Conflict Topologies}
We represent a wireless multi-hop network by a directed connectivity graph $\mathcal{G}^n = (\mathcal{V}, \mathcal{E})$, in which a node $i\in\mathcal{V}$ denotes a transceiver in the network, and a directed link $e = (i, j) \in \mathcal{E}$, where $i, j \in \mathcal{V}$, indicates that packets can be transmitted from node $i$ to node $j$  over the air. 
We assume $\mathcal{G}^n$ to be connected, ensuring that any two nodes within the network can reach each other.
% To describe routing, we denote a directed link as  $(\overrightarrow{ i, j})$, representing data packets being transmitted from node $i$ to node $j$ via wireless link $(i, j)\in\ccalE$.
% Moreover, we consider a set of flows $\mathcal{F}$ within the network. 
% A flow $f = (i, c) \in \mathcal{F}$, where $i \neq c$ and $i, c \in \mathcal{V}$, represents data traffic originating from node $i$ and destined for node $c$, possibly through multiple hops. 

We use a conflict graph, $\mathcal{G}^c = (\mathcal{E}, \ccalH)$, to model the orthogonal access constraint among wireless links as follows: 
each vertex $e \in \mathcal{E}$ corresponds to a wireless link in the network $\mathcal{G}^n$, and the presence of an undirected edge $(e_1, e_2) \in \ccalH$ indicates that links $e_1$ and $e_2$ cannot be activated simultaneously, due to \textit{interface conflict} 
% We primarily focus on the interface conflict model, where the conflict graph is derived from the line graph of the connectivity graph. 
(two links sharing the same radio transceiver), or \textit{radio interference} (e.g., two links located within certain distance interfere with each other if activated at the same time).
% However, our approach also applies to other conflict models, such as the physical distance interference model.
For the rest of this paper, we assume the conflict graph $\mathcal{G}^c$ to be known, either through direct monitoring of the wireless channel by each link or through more advanced estimation as in \cite{zhao2022link, yang2016learning}.

\subsection{Data Plane Architecture}\label{sec:Queuesystem}
We consider a time-slotted medium access control (MAC) system with orthogonal multiple access. 
Each time slot $t$ comprises a resource allocation (routing and scheduling) stage, followed by data transmission.
Vector $\grave{\bbr}(t)=[\grave{\bbr}_{e}(t)\mid e\in\ccalE]\in\reals^{|\ccalE|}$ gathers the real-time link rates in time slot $t$, where
entry $\grave\bbr_{e}(t)$ denotes the number of packets that can traverse link $e\in\ccalE$ during time slot $t$.
% where $\grave{\bbr}_{ij}(t)$ denotes the instantaneous link rate of $(i,j)\in\ccalE$.
Vector $\bbr=\mathbb{E}_{t}[\grave{\bbr}(t)]\in\reals^{|\ccalE|}$ collects the long-term average link rates. 
% \red{The total number of time slots is denoted by $T$.
% The non-negative integer matrix $\mathbf{A} \in \mathbb{Z}_+^{|\mathcal{F}|\times T}$ collects the exogenous packet arrivals in real-time, where $A_{f,t}$ represents the number of packets arriving at the source node of flow $f$ at time slot $t$. 
% The matrix $\mathbf{R} \in \mathbb{Z}_+^{|\mathcal{E}|\times T}$ gathers the stochastic real-time link rates. 
% Each entry $\grave{\bbr}_{e,t}$ denotes the number of packets that can traverse link $e\in\ccalE$ during time slot $t$ and is presumed to exhibit symmetry in both directions.}

The packets destined for node $c\in\ccalV$ are designated as commodity $c$,
$A_i^{(c)}(t)$ denotes the exogenous arrival of commodity $c$ on node $i$ during slot $t$, and the arrival rate $\lambda_i^{(c)}=\mathbb{E}\big[A_{i}^{(c)}\!(t)\big]$.
% $\lambda_i^{(c)}=\mathbb{E}[A_i^{(c)}(t)]$.
The network follows a decoupled routing and scheduling architecture illustrated in Fig.~\ref{fig:timeline}, which abstracts many real-world wireless ad-hoc or mesh networks, and is similar to the Max-Pressure~\cite{varaiya2013max,Levin2023MP} in road networks.
In particular, a commodity-blind Max-Weight link scheduling activates a set of non-conflicting, most congested links in each time slot, for maximum queue stability under an independent, probabilistic packet routing policy.

% \subsection{Per-Neighbor Queueing System}
\vspace{1mm}
\noindent\textbf{Per-neighbor queueing system:}
% In our proposed routing scheme, 
Each device $i\in\ccalV$ hosts per-neighbor queues, denoted as $\ccalQ_{ij}$, which buffers packets to be sent to each neighboring device $j\in\ccalN_{\ccalG^n}(i)$.
In addition, we use queue $\ccalQ_{ii}$ to buffer newly arrived exogenous (injected by user $i$) and endogenous (from neighbors) packets 
% from user $i$ (exogenous) or other devices (endogenous) 
of which the next hops are undecided yet.
All the queues follow the FIFO principle, and the queue length of $\ccalQ_{ij}$ at the beginning of time slot $t$ is denoted by $q_{ij}(t)$.

To facilitate our proposed virtual queue solution and theoretical analysis, we also denote the backlog size of a commodity $c$ on device $i$ with $Q_{i}^{(c)}(t)$, e.g., as packet counts under the data-plane per-neighbor queues.

\vspace{1mm}
\noindent\textbf{Probabilistic forwarding:}
At the beginning of a time slot $t$ and on each node $i\in\ccalV$, every packet from $\ccalQ_{ii}$  is independently sent to queue $\ccalQ_{ij}$ associated with one of its neighbors $j\in\ccalN_{\ccalG^n}(i)$ according to an established per-commodity probabilistic routing policy, $ \big[p^{(c)}_{ij}(t) \mid j\in\ccalN_{\ccalG^n}(i), c\in\ccalV \big]$, where $ 0\leq p^{(c)}_{ij}(t)\leq 1, \sum_{j\in\ccalN_{\ccalG^n}(i)} p^{(c)}_{ij}(t) = 1 $.

\vspace{1mm}
\noindent\textbf{MaxWeight scheduling:}
After forwarding, the utility vector $\bbu(t)\in\reals_{+}^{|\ccalE|}$ is computed, where per-link utility is
\begin{equation}\label{eq:utility}
    % u_{ij}(t)=\max\{q_{ij}(t),q_{ji}(t)\}\grave{\bbr}_{ij}(t)\;.
    u_{ij}(t)=q_{ij}(t)\grave{\bbr}_{ij}(t)\;,\quad \forall (i,j)\in\ccalE.
\end{equation}
% The direction of transmission for link $(i,j)$ is selected accordingly by the $\max$ operation.
Next, Max-Weight scheduling finds the schedule $\bbs(t)$ based on the utility vector $\bbu(t)$ and conflict graph $\ccalG^c=(\ccalE,\ccalH)$, 
\begin{subequations}\label{eq:mwis}
\begin{align}
    & \bbs(t) = \argmax_{\tilde\bbs(t) \in \{0, 1\}^{|\mathcal{E}|}} \bbs(t)^\top \bbu(t),\\
    \text{s.t.}\quad & \tilde s_{e_1}(t) + \tilde s_{e_2}(t)\leq 1,\; \forall\; (e_1, e_2)\in \ccalH \;.
\end{align}
\end{subequations}
Notice that the Max-Weight scheduling in (\ref{eq:mwis}) involves solving a maximum weighted independent set (MWIS) problem on the conflict graph, which is NP-hard \cite{joo2010complexity}. 
Therefore, in practice, (\ref{eq:mwis}) is solved approximately by heuristics, such as centralized greedy maximal scheduler (GMS), distributed local greedy scheduler (LGS) \cite{joo2011local}, and GNN-enhanced LGS \cite{zhao2022link}. 
In this paper, we choose LGS~\cite{joo2011local} as our distributed Max-Weight scheduler for its simplicity. 
The number of packets to be transmitted on link $(i,j)$ in time slot $t$ is 
$$
\mu_{ij}(t) = s_{ij}(t)\min\{q_{ij}(t), \grave{\bbr}_{ij}(t)\}.
$$
The queue lengths $q_{ij}(t)$ are initialized as $q_{ij}(0)=0$ and evolve such that, at each time slot $t$, the queue length $q_{ij}(t)$ increases by the number of packets routed from node $i$ to queue $\ccalQ_{ij}$ under the probabilistic routing policy, and decreases by $\mu_{ij}(t)$ packets transmitted over link $(i,j)$.

\vspace{1mm} 
\noindent\textbf{Route optimization}
under the decoupled routing and scheduling architecture can be stated as follows:
% , we formulate routing as finding an optimal probabilistic forwarding policy:
\begin{Problem}
We seek to establish a probabilistic policy, 
$$ 
P=\big[p^{(c)}_{ij}(t)\mid {(i,j)\in\ccalE}, c\in\ccalV \big]\;,
$$ 
% that can minimize the latency, maximize the network throughput, and/or adapt to network dynamics.    
that forwards packets into per-neighbor FIFO queues, which are scheduled for transmissions by a MaxWeight scheduler defined in~\eqref{eq:utility} and~\eqref{eq:mwis}. 
The goal is to balance multi-commodity traffic across the network for maximum throughput, minimum end-to-end latency, and rapid adaptation to dynamics such as link failures and node mobility.
\end{Problem}
A formal formulation is detailed in Section~\ref{sec:theory}, equations \eqref{eq:P0_policy} and \eqref{eq:P1_mcmcf}.
Unlike routing in wired networks that can be solved efficiently using all-pairs shortest path algorithms or linear multi-commodity min-cost flow (MCMCF) formulations, route optimization in wireless networks is subject to conflict constraints modeled by $\ccalG^c$, e.g., MaxWeight scheduling in~\eqref{eq:mwis} is NP-hard. 
This makes it an NP-hard non-linear MCMCF problem where link capacity and unit cost depend on the decision variables, i.e., flow rate assignments.
As a result, heuristics like ACO routing (Appendix~\ref{app:aco}) are studied.
While bio-inspired ACO routing can operate in a fully distributed manner, it suffers from low convergence, limited adaptivity to node mobility, and lack of mathematical rigor in route optimization.

% Based on the link utility defined in \eqref{eq:utility}, the number of packets transmitted on an activated link $({ i, j})$ is $\min\{q_{ij}(t), \grave{\bbr}_{ij}(t)\}$.

\section{Ant Backpressure}
% BP routing exhibits two primary weaknesses.
% Firstly, the per-commodity queueing system in BP routing can lead to the last-packet problem for short-lived traffic types lacking new arriving packets to drive existing ones through the network~\cite{Alresaini2016bp,ji2012delay}.
% Secondly, the lack of link capacity sharing among different commodities can diminish the utilization efficiency of spectrum resources, particularly in the presence of many short-lived flows. 
% To mitigate these shortcomings, our Ant-BP scheme adopts the per-neighbor queueing system described in Section~\ref{sec:Queuesystem}.

To address the slow convergence of ACO, Ant-BP adopts a per-commodity routing policy similar to \eqref{eq:sgs},
\begin{equation}\label{eq:policy}
p^{(c)}_{ij}(t) = \frac{\rho^{(c)}_{ij}(t)}{\sum_{l\in\ccalN_{\ccalG^n}(i)} \rho^{(c)}_{il}(t)}\;, \quad \forall i,c\in\ccalV\;,   
\end{equation}
while updating the pheromone intensity $\rho^{(c)}_{ij}(t)$ periodically via virtual routing described in Sections~\ref{sec:BP} and~\ref{sec:vrouting}, which remains the same until its next update. 
Ant-BP seeks to address the establishment of pheromone policy in a fully distributed manner with fast convergence.

% consider ``physical routing" and explained as follows. 

% Next, we explain the establishment of a routing policy. 

% \subsection{Per-Commodity Virtual Queuing System}\label{sec:vqueue}

\subsection{Virtual SP-BP for Pheromone Establishment}
\label{sec:BP}

The pheromone in~\eqref{eq:policy} is established by SP-BP routing~\cite{neely2005dynamic,georgiadis2006resource,zhao2023icassp,zhao2023enhanced,zhao2024tmlcn} on a virtual plane, which operates on a virtual (shadow) per-commodity queuing system, denoted as $\{\tilde{\ccalQ}_i^{(c)} ~|~ i, c \in \mathcal{V}\}$, where $ \tilde{\ccalQ}_i^{(c)} $ is the virtual queue hosted on device $i$ for commodity $c$. 
% Besides the physical queuing system, we also adopt a per-commodity virtual queuing system for virtual routing.
A virtual queue $ \tilde{\ccalQ}_i^{(c)} $ only counts the number of packets for its commodity $c$ in the virtual plane without storing any payload, and its queue length at virtual time step $ \tau $ is denoted as $\tilde Q_i^{(c)}(\tau)$. 
Here, $\tau$, rather than $t$, is adopted as virtual routing operates on a time scale different from the physical time slot $t$.

To illustrate the operations of SP-BP, we define biased pressure as $ \tilde U_{ij}^{(c)}(\tau) = \tilde U^{(c)}_i(\tau)-\tilde U^{(c)}_j(\tau) $, where $\tilde U^{(c)}_i(\tau) = \tilde Q^{(c)}_i(\tau) + B^{(c)}_i$, and $B^{(c)}_i \geq 0$ is a queue-agnostic bias representing the shortest path distance from node $i$ to $c$ on (edge weighted) graph $\ccalG^n$. 
The bias field $\{B^{(c)}_i | i,c\in\ccalV\} $ is established through all-pairs-shortest-path (APSP) algorithm on the connectivity graph $\ccalG^n$ with edge weights $\left[ \delta_e |e\in\ccalE \right]$ defined by link features~\cite{zhao2023icassp,zhao2023enhanced,zhao2024tmlcn}.
% e.g., $ \delta_e=\bar{r}r_{\max}/r_e $, where $r_e=\lim_{T\rightarrow\infty}\mathbb{E}_{t\leq T}(\grave{\bbr}_{e}(t)) $ is the long-term link rate for $e\in\ccalE$, $\bar{r}=\mathbb{E}_{e\in\ccalE}(r_e)$, and $ r_{\max}=\max_{e\in\ccalE} r_e $.
% Link duty cycle could be incorporated into $\delta_e$ for high interference density.
To accommodate network mobility, $\{B^{(c)}_i | i,c\in\ccalV\} $ can be updated periodically.

In each time $\tau$, the virtual SP-BP operates in four steps: 
\vspace{1mm}
\noindent
\textbf{Step 1:}
the optimal commodity $c^*_{ij}(\tau)$ for each link $({ i, j})$ is selected as the one with the maximal pressure:
\begin{equation}\label{eq:optimalcomm}
    c^*_{ij}(\tau) = \underset{c \in \ccalV}{\text{argmax}}\; \{ \tilde U^{(c)}_{ij}(\tau)\}\; .
\end{equation}
% which can be based on hop counts \cite{neely2005dynamic, georgiadis2006resource} or more sophisticated edge weights using link features and graph-based machine learning~\cite{zhao2023icassp,zhao2023enhanced,zhao2024tmlcn}.
% In this paper, we adopt the latter. 
\vspace{1mm}
\noindent
\textbf{Step 2:}
the utility of each link $({ i, j})$ is found as~\cite{zhao2024tmlcn}
\begin{subequations}
\begin{align}
    \grave u_{ij}(\tau) = & ~\grave{\bbr}_{ij}(\tau) {w}_{ij}(\tau)\;, \label{eq:util} \\
    w_{ij}(\tau) = & \max\!\left\{\tilde U^{(c^*_{ij}(\tau))}_{ij}\!(\tau), 0\!\right\}{\cdot\mathds{1}\left[\tilde Q_i^{(c_{ij}^*(\tau))}\!(\tau)\!>\!0\right]},\label{eq:backlog}
\end{align}
\end{subequations}
where vector $\grave{\bbr}(\tau)$ collects the virtual real-time link rate of all links at step $\tau$, and $\grave\bbr_{e}(\tau)$ is generated from the same distribution as $\grave{\bbr}_{e}(t)$ to reflect the characteristics of underlying wireless channel.
% vector ${\bbw}(\tau) = [{w}_{ij}(\tau)|(i, j) \in \mathcal{E}]$.  

\vspace{1mm}
\noindent
\textbf{Step 3:}
Max-Weight scheduling finds the schedule $\grave\bbs(\tau) \in \{0, 1\}^{|\mathcal{E}|}$ to activate a set of non-conflicting links as, 
% where the per-link utility is $u_{ij}(\tau) = \grave{\bbr}_{ij}(\tau) {w}_{ij}(\tau)$, 
% with  $\widetilde{w}_{ij} = \max\{w_{ij}(\tau), w_{ji}(\tau)\}$,
\begin{subequations}\label{eq:bp:mwis}
\begin{align}
    & \grave\bbs(\tau) = \argmax_{\tilde\bbs(\tau) \in \{0, 1\}^{|\mathcal{E}|}} \bbs(\tau)^\top \grave\bbu(\tau),  \\
    \text{s.t. }& \tilde s_{e_1}(\tau) + \tilde s_{e_2}(\tau)\leq 1,\; \forall\; (e_1, e_2)\in \ccalH \;.
\end{align}    
\end{subequations}
% and the direction of the link selected by the max function will be recorded for step 4. 
The Max-Weight scheduler is also selected as LGS in~\cite{joo2011local}.

\vspace{1mm}
\noindent
\textbf{Step 4:}
all of the virtual real-time
link rate $\grave{\bbr}_{ij}(\tau)$ of a scheduled link is allocated to its optimal commodity $c^*_{ij}(\tau)$. 
The final transmission assignments of commodity $c \in \mathcal{V}$ on link $({ i, j})$ is
\begin{equation}\label{eq:trans}
\begin{split}
    \mu^{(c)}_{ij}(\tau) =& \min\{\tilde Q_{i}^{(c)}(\tau), \grave{\bbr}_{ij}(\tau)\} \cdot \grave \bbs_{ij}(\tau) \\ 
    & \cdot \mathds{1}\{c = c^*_{ij}(\tau)\} \cdot \mathds{1}\{w_{ij}(\tau) > 0\}\;.
    % \mu^{(c)}_{ij}(\tau) = \begin{cases} \grave{\bbr}_{ij}(\tau), & \text{if } c = c^*_{ij}(\tau), w_{ij}(\tau) > 0, s_{ij}(\tau) = 1, \\ 0, & \text{otherwise.} \end{cases}    
\end{split}
\end{equation}

\vspace{1mm}
\noindent\textbf{Pheromone establishment:}
The pheromone intensity in \eqref{eq:policy} is then established as 
\begin{equation}\label{eq:pheromone}
    \rho^{(c)}_{ij}(\tau) = {\max\{n^{(c)}_{ij}(\tau) - n^{(c)}_{ji}(\tau), 0\}} + \epsilon\;,
\end{equation} 
where the small constant $\epsilon>0$ ensures $ p^{(c)}_{ij}(\tau) = 1/|\ccalN_{\ccalG^n}(i)|$ when the first term in \eqref{eq:pheromone} is zero for all $ l\in\ccalN_{\ccalG^n}(i) $.
In \eqref{eq:pheromone}, $n^{(c)}_{ij}$ is the cumulative number of virtual packets of commodity $c$ that have traveled across link $({ i, j})$ during virtual SP-BP routing,
%and in Section~\ref{sec:BP}. 
% The number of virtual packets of commodity $c$ traversing across link $({ i, j})$ accumulates 
computed as: $ n_{ij}^{(c)}(0)=0 $, and 
$$ 
n_{ij}^{(c)}(\tau) = (1-\varepsilon)\cdot n_{ij}^{(c)}(\tau-1) + \mu^{(c)}_{ij}(\tau), 
$$ 
where the evaporation rate is typically set as zero ($\varepsilon=0$).
Similarly, all virtual queues are initialized as empty, i.e., $\tilde{Q}_{i}^{(c)}(0)=0$ at the beginning of virtual SP-BP.

In virtual SP-BP routing, no packets are actually generated or transmitted, but only the number of virtual packets of each commodity at each node $i\in\ccalV$ and transmitted over each link $({i,j})$ is tracked and exchanged across the network. 
As a result, the data transmission of virtual routing can be compressed into a very short duration, such that tens or hundreds of time steps of virtual SP-BP routing can be finished in a single physical time slot.
This allows our routing policy to be established within a few time slots, as illustrated by the timelines in Fig.~\ref{fig:timeline}.

The required number of time steps for virtual SP-BP routing to establish our routing policy is denoted as $K$ and configured as a system parameter via trial-and-error. 

\subsection{Virtual Traffic Configuration}\label{sec:vrouting}

In virtual routing, each physical traffic demand (e.g., $\lambda_{i}^{(c)}$) is mapped to a virtual packet injection sharing the same source and destination. 
However, their temporal profiles (rates and durations) can be deliberately decoupled. 
Depending on our knowledge of the physical traffic, we configure the virtual plane using one of two strategies: either perfectly replicating the exact physical rates and durations, or generalizing all injections as continuous streaming traffic (e.g., mapping a short-lived physical burst to a persistent virtual stream). 
Moreover, because virtual routing strictly precedes physical transmission, exact slot-by-slot arrival sequences are fundamentally unknowable a priori. 

Therefore, in Ant-BP, virtual packets are generated at the source nodes via a Poisson process matching the configured target rate. 
Because SP-BP route optimization is inherently agnostic to exact arrival processes~\cite{neely2010stochastic}, this implementation simplifies virtual generation while still establishing a robust, high-quality routing policy.

\subsection{Dynamic Network Adaptation}\label{sec:adaptation}
In dynamic networks, link availability fluctuates due to various environmental and physical factors. 
We introduce additional measures to make Ant-BP adaptive to transient link failure and topology-altering network mobility. 
While both disrupt packet transmission, their differing timescales and structural impacts necessitate distinct treatments.

\textbf{Transient Link Failure:}
A transient link failure is a temporary loss of channel capacity on an existing link $e$, causing scheduled transmissions to momentarily fail. 
Despite this disruption, the affected link $e$ remains a valid edge on the connectivity graph $\mathcal{G}^n$ and is expected to recover. 
To capture the impact of interference and channel fading, these failure events are modeled as a stochastic process characterized by Poisson arrivals and random durations.

When a transmission fails, the packets scheduled on the affected outgoing link are reset to an unscheduled state for that time step.
Furthermore, the pheromone level on the failed outgoing link is decayed by a tunable parameter (e.g., $5\%$), discouraging its selection in subsequent slots.

\textbf{Network Mobility:}
Network mobility involves the physical movement of nodes, that can potentially lead to permanent topology changes, i.e., existing links break and/or new links form on the connectivity graph $\mathcal{G}^n$.

To manage the disruptions caused by these topological shifts, our protocol reacts at both the link and packet levels. 
Structurally, broken links are immediately discarded from the routing pool, while newly established links are initialized with a small constant pheromone value $\epsilon>0$. 
If a packet is scheduled on a link that subsequently breaks, its status is immediately reverted to undecided and moved back to $\ccalQ_{ii}$.

However, simply unscheduling packets can lead to severe queuing delays or leave packets stranded at affected nodes. 
To prevent this, we treat commodities stranded at these nodes as virtual sources during the subsequent periodic virtual SP-BP phase to establish new pheromone policies towards their destinations. 
This process enables the system to organically discover new routes for stranded packets, integrate newly formed links, and adapt to the shifted traffic load. 
While our focus is on spatial node movement, this approach extends naturally to mobility events involving node addition or removal,
with the caveat that packets buffered at removed nodes are unavoidably lost.

\section{Theoretical Analysis}\label{sec:theory}

Depending on the network architecture, routing policies process network state information differently. 
Fully dynamic routing, such as SP-BP, utilizes per-commodity queues to make routing and scheduling decisions that are distinctly per-commodity and slot-by-slot, driven by instantaneous queueing states. 
In contrast, Ant-BP decouples route discovery from packet forwarding. It periodically updates a stationary probabilistic forwarding policy, and applies this policy to a constrained per-neighbor FIFO queueing architecture.

\subsection{Fully Dynamic Routing}

To theoretically ground Ant-BP, we first formulate the joint routing and scheduling task under the fully dynamic regime as a stochastic cost minimization problem. While Ant-BP addresses mixed traffic patterns, we restrict this formulation to strictly streaming traffic to preserve the analytical tractability.

Let $Q_i^{(c)}(t)$ be the per-commodity queue length at node $i$.
Assuming that packets reaching their destinations are consumed immediately, i.e., $Q_c^{(c)}(t)=0$, the queue dynamics on a non-destination node $ i\neq c$ follows:
\begin{equation}\label{E:queue}
Q_{i}^{(c)}\!(t+1) = Q_{i}^{(c)}\!(t) - \!\!\!\!\!\sum_{j\in\ccalN_{\ccalG^n}(i)}\!\!\!\! \mu_{ij}^{(c)}\!(t) + \!\!\!\!\!\sum_{j\in\ccalN_{\ccalG^n}(i)} \!\!\!\!\mu_{ji}^{(c)}\!(t) + A_{i}^{(c)}\!(t)   ,
\end{equation}
where $A_i^{(c)}(t)$ is the exogenous arrivals of commodity $c$ on device $i$, and the decision variables at each time step $t$ is the commodity-link rate assignment collected in
$$
\bbM(t)=\left[\mu_{ij}^{(c)}(t)\right]_{(i,j)\in\ccalE,\;c\in\ccalV}\in\bbPi,
$$ 
where $\bbPi$ denotes the feasible solution space under scheduling constraints, e.g., flow conservation, link capacity, and non-conflicting schedules as follows:
\begin{subequations}\label{E:pi}
\begin{align}
\sum_{j\in\ccalN_{\ccalG^n}(i)} \mu_{ij}^{(c)}\!(t) \leq Q_i^{(c)}(t)\;, &\quad \forall c\in\ccalV, (i,j)\in\ccalE, \label{E:pi:conserve}\\
\sum_{c\in\ccalV} \mu_{ij}^{(c)}\!(t) \leq \grave{\bbr}_{ij}(t) \;, &\quad \forall (i,j)\in\ccalE,\label{E:pi:capacity} \\
s_{e_1}(t) + s_{e_2}(t) \leq 1\;, &\quad \forall\; (e_1,e_2)\in\ccalH , \label{E:pi:conflict}
\end{align}
\end{subequations}
where $s_{e}(t)\in\{0,1\}$ indicates if link $e$ is scheduled at step $t$
$$
s_{e}(t) = \mathds{1}\left[\sum_{c\in\ccalV} \mu_{e}^{(c)}(t)\right]\;.
$$
We define the routing cost at time step $t$ as
\begin{equation}\label{eq:slot_cost}
g(t)= \sum_{(i,j)\in\ccalE}\sum_{c\in\ccalV} \mu_{ij}^{(c)}(t)\,B_{ji}^{(c)},
\end{equation}
where $B_{ji}^{(c)} = B_j^{(c)} - B_i^{(c)}$ represents the change in the remaining distance of commodity $c$ (towards its destination $c$) transmitted from $i$ to $j$ (negative values correspond to distance reduction), and $g(t)$ is the change in the total remaining distance of all packets from decisions $\bbM(t) $. 

\vspace{1mm}
\noindent
\textbf{Stochastic Cost Minimization Formulation.}
For a given demand profile $\bblambda$ within the network capacity region, i.e., $ \bblambda=\big[\lambda_i^{(c)}\big]_{i,c\in\ccalV}\in\text{int}(\Lambda)$,
% where $\lambda_i^{(c)}=\mathbb{E}\big[A_{i}^{(c)}\!(t)\big]$, 
the optimal routing and scheduling policy $\pi^*$ is formulated as one that minimizes time-average cost under stability and scheduling constraints:
\begin{subequations}\label{eq:P0_policy}
\begin{align}
\min_{\pi} &\quad \bar g(\pi)\triangleq \limsup_{T\to\infty}\frac{1}{T}\sum_{t=0}^{T-1}\mathbb{E}\!\left[g(t)\mid \pi\right] \label{eq:P0:obj}\\
\text{s.t.}\quad & \{Q^{(c)}_{i}(t)\}\ \text{are strongly stable} ,\label{eq:P0:stability}\\
& \bbM^{\pi}(t)=\left[\mu_{ij}^{(c),\pi}(t)\right]_{(i,j)\in\ccalE,\ c\in\ccalV} \in\bbPi ,\ \forall t,\label{eq:P0:scheduling}\\ 
& \bbM^{\pi}(t)=\pi\left(\bbQ(t),\grave{\bbr}_{*}(t) \mid \lambda, \ccalG^n,\ccalG^c,\ccalI(t\!-\!1)\right),
\end{align}
\end{subequations}
where $ \ccalI(t-1)=\left\{\bbQ(t'),\grave{\bbr}_{*,t'},\bbM^{\pi}(t')\right\}_{t'=0,\dots,t-1} $ encloses the trajectories of past queueing states, link rates, and decisions.
Notice that $\ccalI(t-1)$ is optional for policy $\pi$.

\begin{proposition}
For Problem~\eqref{eq:P0_policy} in the fully dynamic regime, the SP-BP policy $\pi^s$ is near-optimal for any fixed admissible demand profile $ \lambda \in\text{int}(\Lambda)$.
\end{proposition}

\begin{proof}
% \vspace{1mm}
% \noindent
% \textbf{SP-BP policy:}
Let $U_i^{(c)}(t)= Q_i^{(c)}(t) + B_i^{(c)}$ denote biased backlog,   $U_{ij}^{(c)}(t) = U_i^{(c)}(t) - U_j^{(c)}(t)$ as biased pressure, and $\bbQ(t)=\big[Q_i^{(c)}(t)\big]_{i,c\in\ccalV}$ as the queueing state at step $t$. 
We define one-step conditional drift of the Lyapunov function as:
$$
\Delta(t)=\mathbb{E}[\ccalL(t+1)-\ccalL(t)\mid \bbQ(t)],\; \ccalL(t)=\frac{1}{2}\sum_{i,c\in\ccalV}Q_i^{(c)}(t)^2.
$$
The biased drift $\tilde\Delta(t)$ can be defined similarly based on biased backlogs $\{U_i^{(c)}(t)\}$.
Based on Lyapunov theory for biased BP~\cite[Theorem 1]{zhao2024tmlcn} and its drift-plus-penalty (DPP) interpretation~\cite{zhao2026ton}, the policy of SP-BP ($\pi^s$) is formulated as solving a MaxWeight problem, which minimizes the upper bound of DPP objective at each time step $t$:
\begin{subequations}
\begin{align}
    \bbM^{\pi^s}(t)&=\argmax_{\bbM(t)\in\bbPi} \sum_{(i,j)\in\ccalE}\sum_{c\in\ccalV} \mu_{ij}^{(c)}(t) U_{ij}^{(c)}(t) \label{eq:maxweight}\\
    &= \argmin_{\bbM(t)\in\bbPi}\; \text{UB}\left[\tilde\Delta(t)\right]\;\\
    &= \argmin_{\bbM(t)\in\bbPi}\; \text{UB}\left[\Delta(t) + g(t)\right]\;,\label{eq:dpp}
\end{align}
\end{subequations}
where $\text{UB}(\cdot)$ stands for the upper bound.
It has been proven in~\cite[Theorem 1]{zhao2024tmlcn} that~\eqref{eq:maxweight} guarantees strong stability for all $\bblambda\in\text{int}(\Lambda)$, defined as~\cite{neely2005dynamic,georgiadis2006resource,neely2010stochastic,jiao2015virtual}
\begin{equation}\label{eq:stability}
\limsup_{T\rightarrow \infty}\frac{1}{T}\sum_{t=0}^{T-1}\sum_{i,c\in\ccalV}\mathbb{E}\left[Q_i^{(c)}(t)\right]<\infty\;.    
\end{equation}
Since minimizing the drift bound $\text{UB}[\Delta(t)]$ at each step $t$ guarantees strong stability of queues for demands $\bblambda\in\text{int}(\Lambda)$~\cite{neely2005dynamic,georgiadis2006resource}, by the standard DPP framework~\cite[Theorem 4.8]{neely2010stochastic}, the per-slot drift bound minimization in~\eqref{eq:dpp} yields a policy that achieves near-optimal time-average cost $\bar g(\pi)$ and queue stability in~\eqref{eq:P0_policy}, with the tradeoff governed by the scaling of the bias field $\{B_{i}^{(c)}\}$.
Therefore, the DPP minimization of SP-BP in~\eqref{eq:dpp} can be interpreted as an online Lagrangian for Problem~\eqref{eq:P0_policy}. With a properly scaled bias field~\cite{zhao2024tmlcn}, SP-BP thus serves as a near-optimal policy for Problem~\eqref{eq:P0_policy} with any demands $\bblambda\in\text{int}(\Lambda)$.
% This provides an online primal-dual interpretation of SP-BP in~\eqref{eq:dpp} for Problem~\eqref{eq:P0_policy}. SP-BP with a properly scaled bias field~\cite{zhao2024tmlcn} thus yields a near-optimal policy for Problem~\eqref{eq:P0_policy}.
\end{proof}

This optimality gap arise from the bias field scaling tradeoff in \eqref{eq:dpp} and the tightness of the drift upper bound when a heuristic MaxWeight scheduler is applied.
In practice, optimality is further limited by the finite virtual routing horizon, node mobility, and virtual-physical mismatches in both demand profiles and instantaneous link rates.

% The nearness of optimality depends on several factors:
% theoretically, the tradeoff governed by the scaling of bias fields in~\eqref{eq:dpp}, the tightness of the upper bound with heuristic MaxWeight scheduler; 
% practical limitations like finite horizon of virtual routing, mismatched virtual-physical demand profile, and knowledge and realization of instantaneous link rates, and mobility. 

\subsection{Stationary Policy in Restricted Regime}

\begin{definition}
For any stabilizing policy $\pi$, define the induced time-average (expected) commodity flow rates for an arbitrary admissible demand profile $\lambda\in \text{int}(\Lambda)$ as:
\begin{equation}\label{eq:avg_flow_def}
f_{ij}^{(c),\pi}\triangleq \limsup_{T\to\infty}\frac{1}{T}\sum_{t=0}^{T-1}\mathbb{E}\big[\mu_{ij}^{(c),\pi}(t) \big],\quad \forall (i,j)\in\ccalE,c\in\ccalV .
\end{equation}    
\end{definition}
This induced time-average commodity flows represents the long-term behavior of a stabilizing policy under $\bblambda$. 

\begin{lemma}\label{L:fcons}
The time-average commodity flows induced by a stabilizing policy satisfy flow conservation.
\end{lemma}
\begin{proof}
Applying strong stability of queues in \eqref{eq:stability} and instantaneous flow conversation in \eqref{E:pi:conserve} to the telescopic sum of queue dynamics in~\eqref{E:queue}  leads to flow conservation 
$$
\sum_{j\in\ccalN_{\ccalG^n}(i)} f_{ij}^{(c),\pi} - \!\!\!\!\sum_{j\in\ccalN_{\ccalG^n}(i)} f_{ji}^{(c),\pi} = \lambda_{i}^{(c)}\;,\quad \forall i,c\in\ccalV\;.
$$
With node-edge incidence matrix $\bbA$ of graph $\ccalG^n$, we have:
$$
\bbA \bbf^{(c),\pi}  = \bbd^{(c)} ,\quad\bbf^{(c),\pi} = \left[f_{ij}^{(c),\pi}\right]_{(i,j)\in\ccalE},\quad \forall c\in\ccalV,
$$
where $\bbd^{(c)}_i=\lambda_i^{(c)}$ and $\bbd_c^{(c)}=-\sum_{i\neq c}\lambda_i^{(c)}$.    
\end{proof}

\begin{proposition}
Under a policy regime defined by stationary probabilistic forwarding and per-neighbor FIFO queues, Problem~\eqref{eq:P0_policy} is equivalent to a constrained multi-commodity min-cost flow (MCMCF) problem for demand $\bblambda \in\text{int}(\Lambda')$:
\begin{subequations}\label{eq:P1_mcmcf}
\begin{align}
& \min_{\pi}\quad \sum_{(i,j)\in\ccalE}\sum_{c\in\ccalV} f_{ij}^{(c),\pi}\cdot B_{ji}^{(c)} \label{eq:P1:obj} \\
\text{s.t.}\quad 
% & \{Q_{ij}(t)\}_{(i,j)\in\ccalE}\ \text{are strongly stable},\\
& \bbA \bbf^{(c),\pi}  = \bbd^{(c)},\quad \forall c\in\ccalV, \label{eq:P1:stability} \\
& \bbM^{\pi}(t) \in\bbPi ,\ \forall t, \label{eq:P1:schedule} \\
& \bbM^{\pi}(t) = \pi^a \left( \grave{\bbr}_{*}(t) \mid \lambda, \ccalG^n,\ccalG^c,\bbF^{\pi}\right), \label{eq:P1:forward}
\end{align}
\end{subequations}
where matrix $\bbF^{\pi} = \left[f_{ij}^{(c),\pi}\right]_{(i,j)\in\ccalE, c\in\ccalV}$ collects all commodity flows.
\eqref{eq:P1:forward} defines the restricted policy regime: probabilistic forwarding, per-neighbor FIFO queues, and commodity-blind MaxWeight scheduling in \eqref{eq:mwis}.
\end{proposition}
\begin{proof}
Apply time average to \eqref{eq:slot_cost}, and by linearity of expectation, the objectives in \eqref{eq:P0:obj} and~\eqref{eq:P1:obj} are identical
$$
\bar g(\pi)=\sum_{(i,j)\in\ccalE}\sum_{c\in\ccalV} f_{ij}^{(c),\pi}\cdot B_{ji}^{(c)}\;.
$$ 
Based on Lemma~\ref{L:fcons}, constraint~\eqref{eq:P1:stability} is equivalent to~\eqref{eq:P0:stability}. 
Constraints~\eqref{eq:P1:schedule} and~\eqref{eq:P0:scheduling} are identical.
\end{proof}

\begin{remark}
Ant-BP approximates the optimal policy for Problem~\eqref{eq:P1_mcmcf} in the restricted regime.
\end{remark}
\noindent
Ant-BP constructs a pheromone routing policy in~\eqref{eq:pheromone} through virtual SP-BP in~\eqref{eq:optimalcomm}--\eqref{eq:trans}, and therefore heuristically approximate the constrained MCMCF problem in~\eqref{eq:P1_mcmcf} under restricted per-neighbor FIFO queue architecture. 
However,  since the MaxWeight objective in~\eqref{eq:maxweight} can no longer be optimized slot-by-slot under this restricted regime, the stability region of Ant-BP is generally reduced.

% Therefore, our Ant-BP can be interpreted as a heuristic routing solution targeting~\eqref{eq:P0_policy} under a per-neighbor FIFO queue system.

Through experiments in Section~\ref{sec:results}, we demonstrate that the pheromone policy established from virtual streaming traffic generalizes well to scenarios with short-lived traffic and mismatched traffic loads.

\section{Numerical Experiments}\label{sec:results}
We evaluate the effectiveness, robustness, and trade-offs of Ant-BP in both route optimization and last-packet problem mitigation, by comparing it against state-of-the-art SP-BP~\cite{zhao2024tmlcn} and ACO routing under challenging scenarios, such as mixed streaming and bursty traffic, uncertain arrival rates, transient link failures, and mobility.

% To assess the performance of our proposed Ant-BP under link failure and mobility, we conducted evaluations in simulated wireless multi-hop networks. These evaluations included comparisons with several benchmarks, including the 

\subsection{Test Setup and Baselines}\label{sec:results:config}
We simulate wireless ad-hoc networks of $|\ccalV|=100$ nodes generated from a 2D point process with a uniform density of $8/\pi$ in a square area.
Assuming beamforming antennas and uniform transmit power, we adopt a unit-disk graph model, in which two nodes are connected if their distance $\leq 1$.
This unitless graph abstraction is frequency-agnostic and ensures broad applicability across various physical layer specifications.
By mitigating radio interference, this setup yields a moderate network density with an average conflict graph degree of $13.86$.
% While this configuration most closely models conventional wireless ad hoc networks. 
% we additionally evaluate alternative topologies, including Barabási–Albert graphs with $m \in \{1,2\}$ and a hierarchical server–user structure to represent computation offloading scenarios~\cite{erfaniantaghvayi2025selr}.
% \red{SS: [Add results for Barabási–Albert and hierarchical topologies here, or remove this claim if results will not be included.]}

We generate $100$ test instances derived from $10$ random topologies, each paired with $10$ random realizations of source-destination pairs and stochastic link rates. 
Each instance features a uniformly random number of flows between $\lfloor 0.15 |\ccalV|\rfloor$ and $ \lceil 0.30|\ccalV|\rceil$.
To capture fading channels with lognormal shadowing, OFDM-like wide-band waveforms, and power control, long-term link rates are drawn from a uniformly distribution $\bbr_{e} \sim \mathbb{U}(10,42)$, while real-time link rates fluctuate as $\grave{\bbr}_{e}(t) \sim  \mathbb{N}(\bbr_{e}, 3^2)$, truncated to $\bbr_e\pm 9$.
Simulations run for $T=1000$ time steps.\footnote{The implementation details and source code for this study are available at \url{https://github.com/Negar-Erfanian/AntBP.git}.
}

\begin{figure*}[t!]
    \hspace{-3mm}
    \subfloat[]{
    \includegraphics[width=0.32\linewidth]{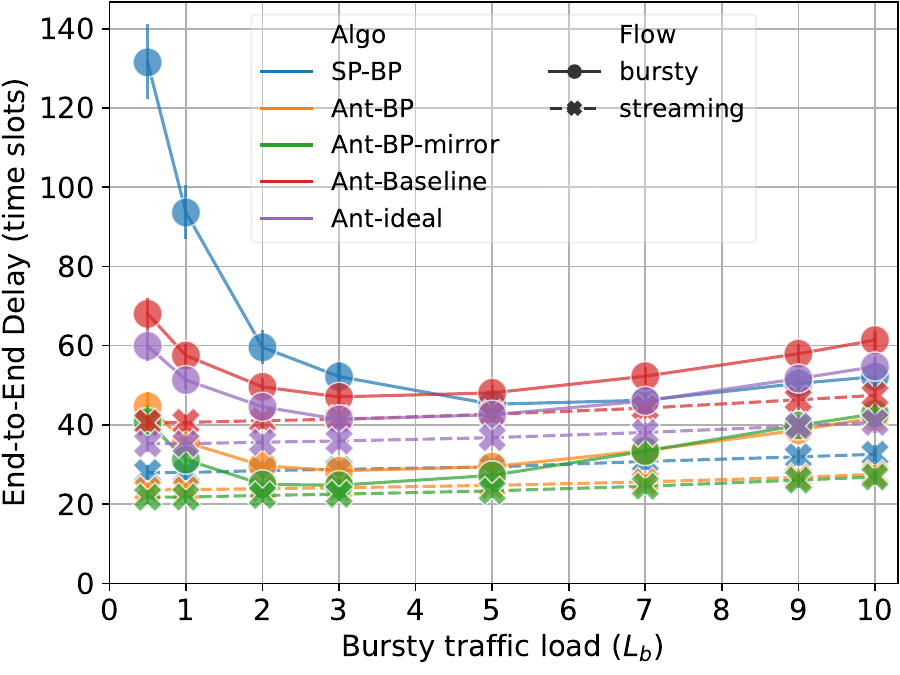}
    \label{fig:mixed:delay}
    }
    \hspace{-3mm}
    \subfloat[]{
    \includegraphics[width=0.32\linewidth]{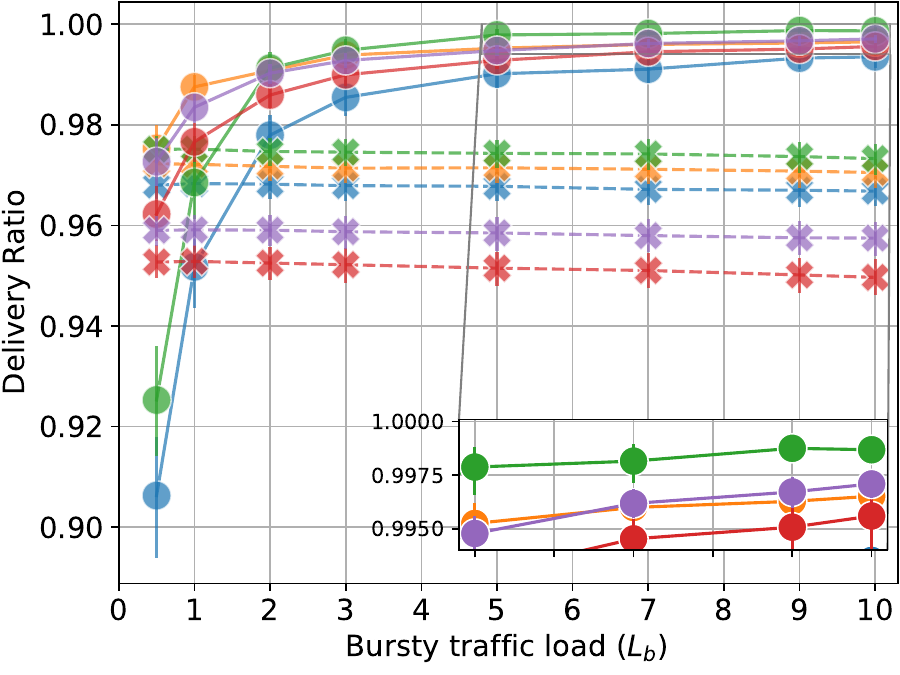}
    \label{fig:mixed:delivery}
    }
    \hspace{-3mm}
    \subfloat[]{
    \includegraphics[width=0.32\linewidth]{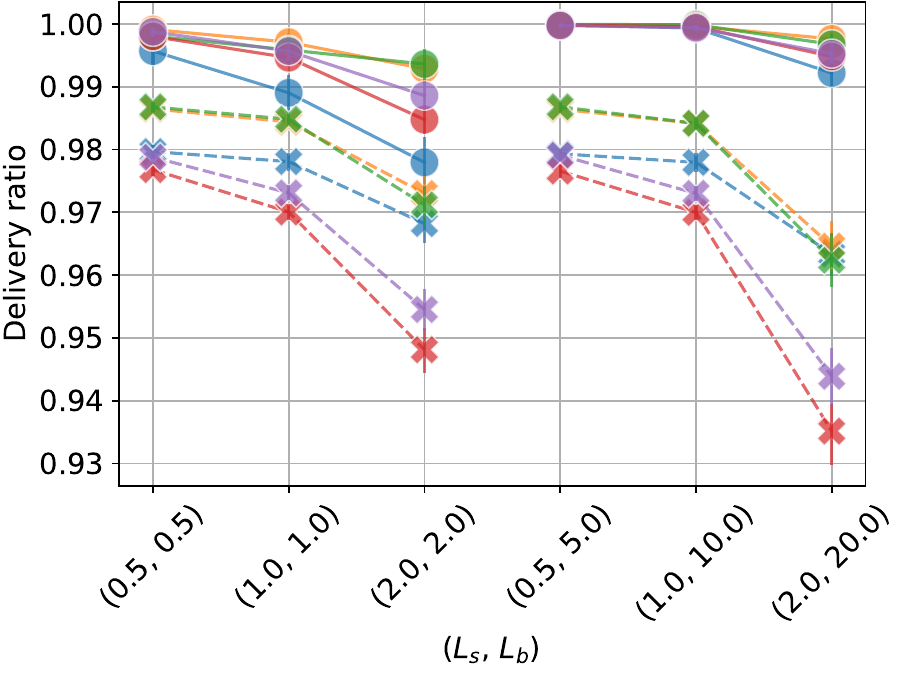}\label{fig:robust}\label{fig:robust:delivery}
    }  
    \vspace{-0.1in}
    \caption{The average (a) end-to-end latency and (b) delivery ratio of tested routing policies as a function of bursty flow load $ L_b $ in 10 random instances of wireless networks of 100 nodes under mixed traffic setting with a constant streaming load  $L_s=2.0$. 
    (c) The average delivery ratio of tested routing schemes in wireless networks of 100 nodes under mixed traffic with different loads shown in the x-axis. 
    For Ant-based schemes, routing policies are established based on virtual traffic flows of $(1.0, 1.0)$ and $(1.0, 10.0)$ for the left and right sides, respectively.}
    \label{fig:mixed}
    \vspace{-0.1in}
\end{figure*}

We define two traffic types: streaming and bursty, with exogenous packet arrivals following a Poisson process.
Streaming flows have a constant rate $L_{s}\lambda_{s} $ throughout the simulation, where $\lambda_{s}\in \mathbb{U}(0.2, 1.0)$ and $L_s$ is the streaming load.
Bursty flows operate at a rate $L_b\lambda_b$ (with $\lambda_{b}\in \mathbb{U}(0.2, 1.0)$ and load $L_b$) strictly within a short active window $[t_s, t_s + 30]$, and $\lambda_b = 0$ otherwise. 
The burst start time $t_s$ is drawn uniformly from $[0, T-100]$ for physical routing, and $t_s=0$ for virtual routing.
In mixed traffic settings, a flow is configured as bursty with probability $P_{b}=0.5$.

We evaluate two variations of Ant-BP against three baselines.
For all policies utilizing shortest path biases, the biases $\{B_{i}^{(c)}|i,c\in\ccalV\}$ are established via APSP on $\ccalG^n$ with edge weights $ \delta_e=\bar{r}r_{\max}/r_e, \forall e\in\ccalE$ \cite[Sec. IV-B]{zhao2024tmlcn}. 
All virtual routing phases run for $K=1000$ time steps.
\begin{enumerate}[leftmargin=1em,labelsep=0.2em,itemsep=0pt,topsep=1pt]
    \item \textit{Ant-BP}: Virtual flows are configured strictly as streaming with load $L_s$, ignoring actual physical traffic profiles.
    \item \textit{Ant-BP-mirror}: Virtual flows perfectly replicate the rate and traffic type of their corresponding physical flows.
    \item \textit{SP-BP}: State-of-the-art SP-BP~\cite{zhao2023enhanced,zhao2024tmlcn} running directly on physical traffic with per-commodity queues.
    \item \textit{Ant-Baseline}: Routing policy established via ACO routing (Appendix~\ref{app:aco}) in virtual phase and frozen during physical routing. Pheromones are initialized as $\rho^{(c)}_{ij}(0) = 1.3$ and updated per time step by a constant $\theta_{ij,k}^{(c)}(\tau) = 0.01$ (with evaporation $\varepsilon=0.002$). 
    % for $m$ ants traversing link $(i,j)$. 
    To align with Ant-BP, virtual routing probabilities follow \eqref{eq:sgs}:
    \begin{equation}
    \label{eq:sgs:test}
    p_{ij}^{(c)}(\tau) = \frac{\rho_{ij}^{(c)}(\tau) + B_{ij}}{\sum_{l\in\ccalN_{\ccalG^n}(i)} \left[\rho^{(c)}_{il}(\tau) + B_{il}\right]}\;,
    \end{equation}
    where $B_{ij}=B_{i}\!-\!B_j$ (Sec.~\ref{sec:BP}) acts as a heuristic link cost.

    \item \textit{Ant-Ideal}: Pheromones are initialized like Ant-Baseline but actively maintained during physical routing by proactive ants (sent per 100 data packets, with a 10\% uniform exploration probability). The path cost $\phi\big(\ccalP^{(c)}_k\big)$ in \eqref{eq:ph_update:delta} is ant $k$'s end-to-end latency, and pheromones update (see \eqref{eq:ph_update}) instantaneously upon delivery, without backward ant propagation.     
\end{enumerate}
Crucially, despite sharing the same $K$ steps of virtual routing, Ant-BP incurs significantly lower physical overhead than Ant-Baseline and Ant-Ideal, as it only exchanges packet counts rather than sending out actual ant packets.

\subsection{Performance under Mixed Traffic and Injection Rate Uncertainties}\label{sec:results:mixed}
To evaluate routing quality and bursty-flow handling under mixed traffic, we test all schemes with a moderate fixed streaming load $L_s = 2.0$ and varying bursty loads $L_b \in \{0.5, 1, \dots, 10\}$.
% Because bursty packets constitute a small fraction of the total traffic ($L_b\cdot30/T$), 
Because bursts occupy only a small fraction of the total time horizon ($30/T$), increasing the bursty load $L_b$ contributes only a limited amount of additional traffic. As a result, the synthetic latency\footnote{By default, latency of undelivered packets is treated as $T$, except for Section~\ref{sec:results:mobility}, where it is residency time ($T-$injection time)} and delivery ratios of streaming flows remain stable as $L_b$ increases, as illustrated by dashed lines in Figs.~\ref{fig:mixed:delay} and~\ref{fig:mixed:delivery}.
% Given a constant $L_s$, as $L_b$ increases, the total number of packets injected into the networks increases slightly 
% therefore, the average latency (delivery ratio) of streaming flows for all routing schemes only increases (decreases)  slightly, 
% The average end-to-end latency and delivery ratio as a function of $L_b$ are shown in Figs.~\ref{fig:mixed:delay} and~\ref{fig:mixed:delivery}, respectively.
For streaming flows, Ant-BP-mirror and Ant-BP achieve the best and second-best latency and delivery ratios, followed by SP-BP, Ant-Ideal, and Ant-Baseline. 
Their respective average delivery ratios are $0.974, 0.971, 0.968, 0.958$, and $0.952$. 
% Their average latencies are $23.62, 24.97, 29.72, 37.16$, and $43.04$, respectively.

The latency of bursty flows is consistently higher and more sensitive to $L_b$ than streaming flows (solid lines in Fig.~\ref{fig:mixed:delay}), as bursty traffic lacks the persistent congestion gradients to drive queue-based Max-Weight scheduling.
However, because bursty packets arrive entirely within the a window of $30$ time slots starting at $t_s$, most of them have sufficient time to reach their destinations within the simulation horizon $T=1000$, resulting in higher delivery ratios than streaming traffic (Fig.~\ref{fig:mixed:delivery}).
Ant-BP and Ant-BP-mirror consistently outperform the ACO baselines across both metrics, demonstrating the superiority of virtual SP-BP pathfinding over traditional ACO.
Furthermore, their performance gains over SP-BP highlight the effectiveness of per-neighbor FIFO queues and link capacity sharing in mitigating the last-packet problem.

% As illustrated by the solid lines in Fig.~\ref{fig:mixed:delay}, the latency of bursty flows is always higher than that of streaming flows, and more sensitive to $L_b$,  
% since bursty traffic lacks consistent congestion gradients to drive the scheduling of links on its routes in the presence of streaming flows under queue length-based Max-Weight scheduling.
% However, since all packets of bursty flows arrive in the first $30$ time slots of the simulation, most of them reach their destinations after spending sufficient time in the network $T=1000$, resulting in a higher delivery ratio than streaming traffic, as illustrated in Fig.~\ref{fig:mixed:delivery}.
% The fact that Ant-BP and Ant-BP-mirror almost always outperform Ant-Baseline and Ant-Ideal in latency and delivery ratio for streaming and bursty flows shows that the pathfinding based on virtual SP-BP is superior to ACO. 
% The gain of Ant-BP and Ant-BP-mirror over SP-BP also demonstrates the effectiveness of FIFO-based link capacity sharing among commodities.

Lower burst loads ($L_b \leq 3$) present the most challenging starvation regime, where performance degrades for all schemes as $L_b$ decreases.
SP-BP suffers the most severe degradation (e.g., a latency of $131.5$ and delivery ratio of $0.906$ at $L_b=0.5$) because of its exclusive commodity selection based on instantaneous congestion gradients.
Conversely, Ant-BP avoids this starvation by allowing different commodities to share link capacity and per-neighbor queues in a FIFO fashion, achieving the best delivery ratio ($0.975$) and second-best latency ($44.7$) at $L_b=0.5$.
Ant-BP-mirror ranks best overall for $L_b\geq 2$ by perfectly leveraging exact bursty flow knowledge.
However, at $L_b=0.5$, its delivery ratio drops significantly to $0.925$ because its exact virtual traffic replica also suffers the last-packet problem.
By configuring all virtual flows as streaming, Ant-BP eliminates this issue, managing lightweight bursty traffic highly effectively at the cost of slight performance degradation under heavier loads compared to Ant-BP-mirror.

% For lower burst loads $L_b \leq 3$, the latency and delivery ratio of bursty traffic for all routing schemes get worse as $L_b$ decreases, as illustrated in Figs.~\ref{fig:mixed:delay} and~\ref{fig:mixed:delivery}, presenting the most challenging cases where bursty traffic is starved from link scheduling. 
% In particular, among all the schemes, SP-BP exhibits the worst latency (e.g., $131.5$ for $L_b=0.5$) and the lowest delivery ratio (e.g., $0.906$ for $L_b=0.5$) under bursty traffic conditions, due to its routing decisions also depending on congestion gradients.
% Ant-BP achieves the best delivery ratio (e.g., $0.975$ for $L_b=0.5$) and ranks second in latency (e.g., $44.7$ for $L_b=0.5$) under bursty flows with $L_b \leq 3$. 
% It does so by leveraging the superior pathfinding capability of virtual SP-BP and FIFO-based link capacity sharing.
% Ant-BP-mirror ranks the best in delivery ratio and latency for $L_b\geq 2$ by further leveraging the exact knowledge of bursty flows. 
% However, for $L_b=0.5$, its delivery ratio drops significantly to $0.925$ since its virtual SP-BP also suffers from the last-packet problem.
% In contrast, Ant-BP can avoid this issue by configuring all the virtual flows as streaming, 
% thus, it manages lightweight bursty traffic more effectively, at the cost of a slight degradation in latency and delivery ratio for streaming and heavier bursty traffic compared to Ant-BP-mirror.

To evaluate robustness against imperfect physical flow rate estimation, we test the ant-based schemes under physical flow rates that are half or double those used for route establishment ($L_s=1.0, L_b=1.0$ and $L_s=1.0, L_b=10.0$).
We compare these against SP-BP operating with exact instantaneous queueing state information.
Despite these mismatches, the delivery ratio trends and scheme rankings (Fig.~\ref{fig:robust}) remain consistent with the exact-knowledge scenario (Fig.~\ref{fig:mixed:delivery}), with Ant-BP and Ant-BP-mirror maintaining their lead.
In the extreme heavy-traffic case ($L_s=2.0, L_b=20.0$), SP-BP achieves slightly lower latency for both streaming ($41.2$ vs. $47.6$ and $44.1$) and bursty flows ($79.8$ vs. $92.3$ and $82.1$) due to its slot-by-slot management.
Nonetheless, its delivery ratios remain similar for streaming traffic and slightly worse for bursty flows compared to the Ant-BP variants.
These results confirm that Ant-BP and Ant-BP-mirror routing policies are highly robust to mismatched virtual flow configurations.

% Virtual routing may face the challenge of imperfect knowledge of physical traffic flow rates. 
% To evaluate the robustness of routing policies learned from mismatched flow rates, we test four ant-based schemes under physical flow rates that are half or double of those used for route establishment ($L_s=1.0, L_b=1.0$ and $L_s=1.0, L_b=10.0$). 
% For comparison, we also test SP-BP based on the exact queueing state information.
% The overall trends of delivery ratio by traffic load and the rankings of tested schemes in Fig.~\ref{fig:robust} are consistent with the results in Fig.~\ref{fig:mixed:delivery}. 
% Even with mismatched virtual flow rates, Ant-BP and Ant-BP-mirror generally maintain their leading ranks for both streaming and bursty traffic. 
% In the special case of $L_s=2.0, L_b=20.0$, SP-BP achieves slightly lower latency than Ant-BP and Ant-BP-mirror for both streaming flows ($41.2$ v.s.  $47.6$ and $44.1$) and bursty flows ($79.8$ v.s. $92.3$ and $82.1$) due to its advantage of managing heavier traffic; 
% nonetheless, its delivery ratios are similar to those of Ant-BP and Ant-BP-mirror in streaming traffic ($0.963$ v.s. $0.960$ and $0.962$) and slightly worse in bursty flows ($0.992$ v.s $0.995$ and $0.997$).
% It shows that the routing policies of Ant-BP and Ant-BP-mirror are robust to mismatched virtual flow rate configurations.

\FloatBarrier
\subsection{Goodput under Streaming Stress Tests}\label{sec:results:throughput}
We evaluate network goodput (end-to-end throughput defined as total delivered packets per time slot) of four schemes under a pure streaming traffic setting (Ant-BP-mirror and Ant-BP are equivalent in this setting), across loads $L_s \in \{0.5, 1,2, \dots, 12\}$.
As shown in Fig.~\ref{fig:throughput}, Ant-BP achieves similar or slightly higher goodput than SP-BP at lower loads ($L_s \leq 3$), but averages only $84.4\%$ of the goodput of SP-BP for $L_s > 3$.
This divergence occurs because Ant-BP's link capacity sharing yields limited benefits for uniform streaming traffic, whereas the instantaneous queue-based resource allocation of SP-BP excels at managing heavy congestion.
Combined with earlier results, this confirms that Ant-BP significantly improves bursty traffic handling without compromising goodput under low-to-medium streaming intensities.
Additionally, Ant-Ideal surpasses Ant-BP in goodput at extreme loads ($L_s \geq 10$) by leveraging cost-based pheromone deposits and continuous proactive route maintenance.
This suggests that incorporating dynamic route maintenance into Ant-BP could further enhance its heavy-traffic scalability.

\begin{figure}[!t]
    \centering
    \includegraphics[width=0.9\linewidth]{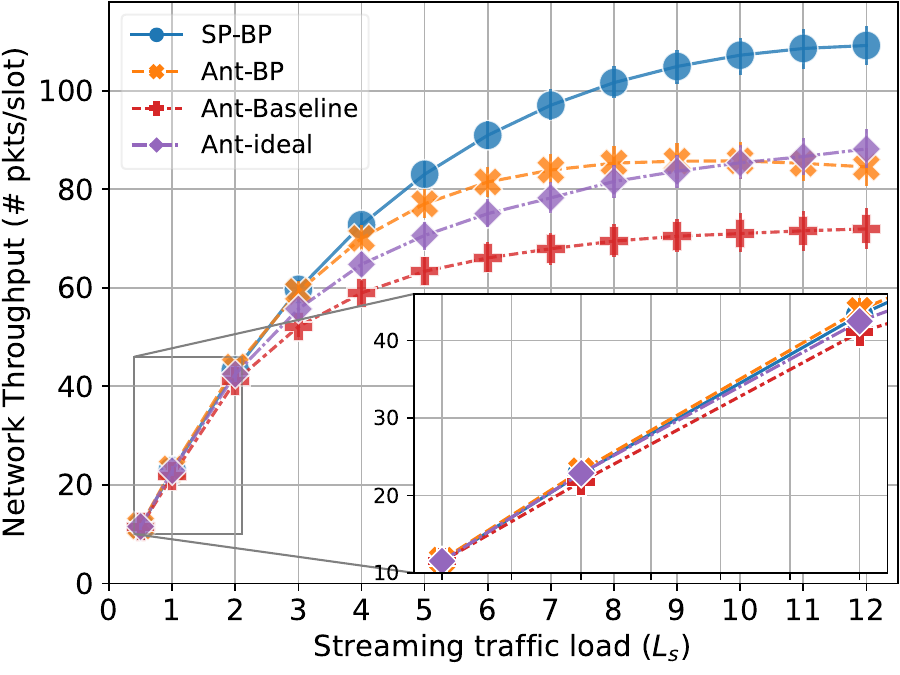}
    \caption{Average network goodput (network-wide number of packets delivered per time slot) of routing schemes on wireless networks with 100 nodes under all streaming traffic as a function of traffic load $L_s$.
    }
    \label{fig:throughput}
\end{figure}

% To better understand the pros and cons of Ant-BP, we compare the throughput performance of the routing schemes under all streaming flows across different traffic loads $L_s \in \{0.5, 1, 2, \dots, 12\}$, by counting the number of delivered packets per physical time-slot across the network. 
% In Fig.~\ref{fig:throughput}, Ant-BP achieves similar or slightly better throughput compared to SP-BP for lower traffic loads ($L_s\leq3$), but only $84.4\%$ of throughput of SP-BP for $L_s>3$ on average. 
% This is because the link capacity sharing of Ant-BP brings limited benefits for streaming traffic, whereas the resource allocation decisions of SP-BP based on instantaneous queueing state information can better manage heavier traffic.
% This result, along with observations from Fig.~\ref{fig:mixed}, demonstrates that Ant-BP can improve the performance of bursty traffic over SP-BP without compromising throughput under low-to-medium streaming traffic intensities ($L_s\leq 3$). 
% Moreover, the throughput of Ant-Ideal surpasses that of Ant-BP for higher traffic loads $L_s\geq 10$ due to its cost-based pheromone deposit mechanism and continuous route maintenance using proactive ants. 
% It suggests that the throughput of Ant-BP might be further improved by incorporating such mechanisms of Ant-Ideal.

\FloatBarrier
\begin{figure*}[!t]
    % ---------- Row 1 ----------
    \hspace{-3mm}
    \subfloat[]{
        \includegraphics[width=0.32\linewidth]{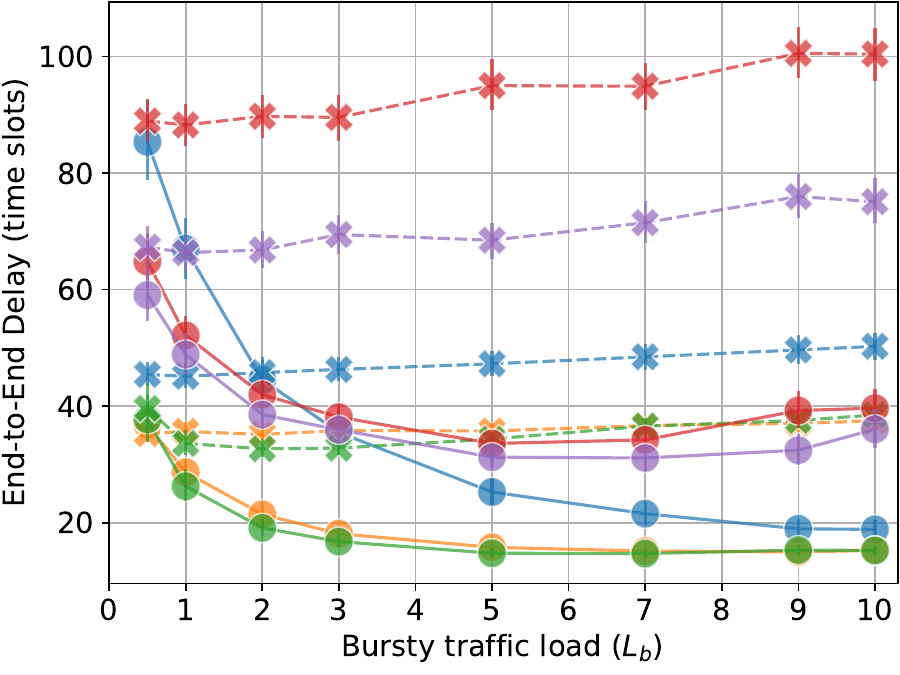}
        \label{fig:mixed:delayalllink}
    }
    \hspace{-3mm}
    \subfloat[]{
        \includegraphics[width=0.32\linewidth]{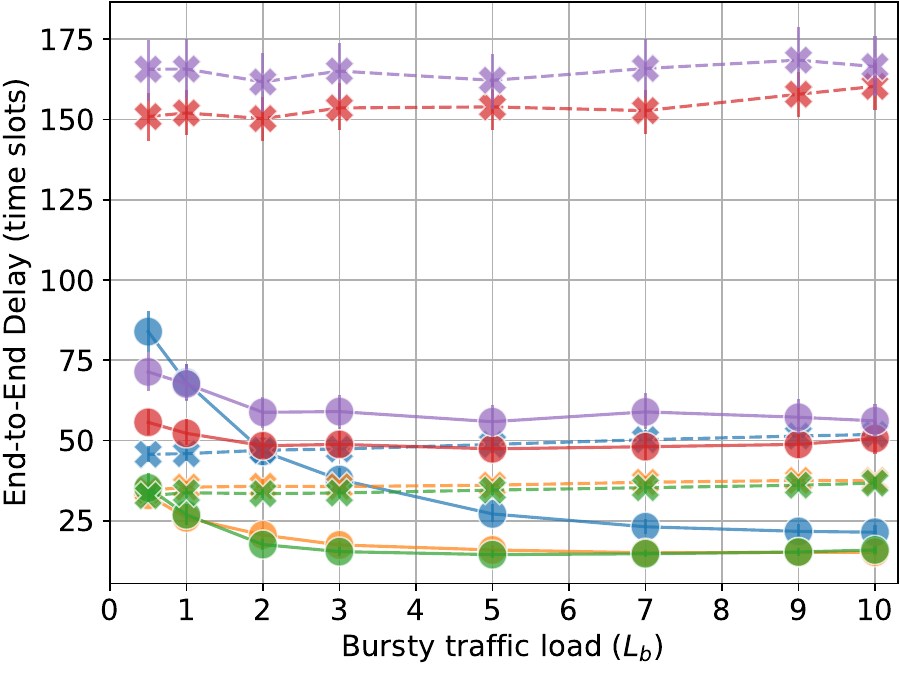}
        \label{fig:mixed:delaybw}
    }
    \hspace{-3mm}
    \subfloat[]{
        \includegraphics[width=0.32\linewidth]{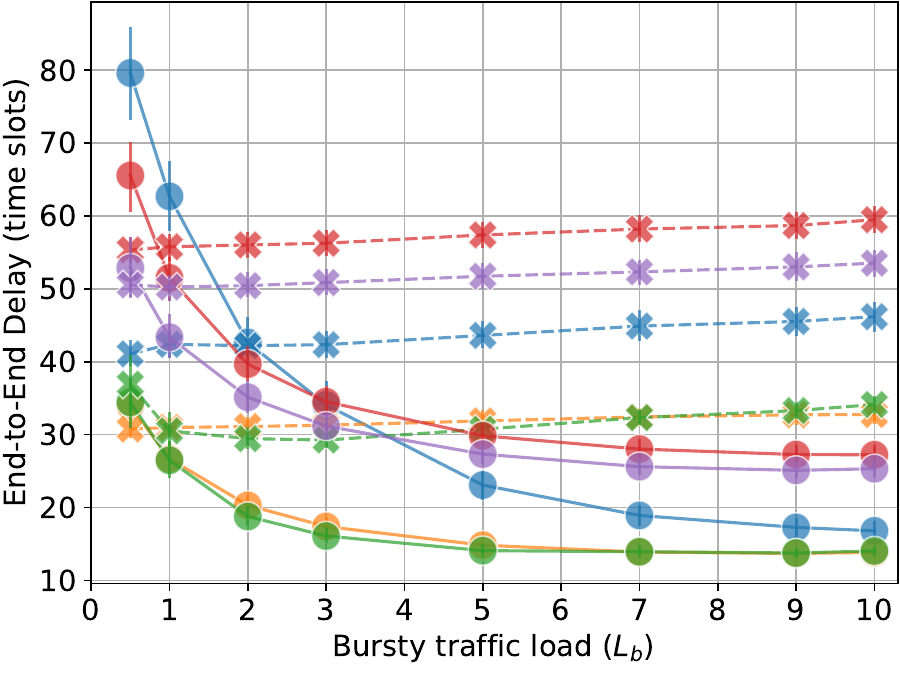}
        \label{fig:mixed:delaylocal}
    }

    % Small vertical space between rows
    \vspace{1mm}

    % ---------- Row 2 ----------
    \hspace{-3mm}
    \subfloat[]{
        \includegraphics[width=0.32\linewidth]{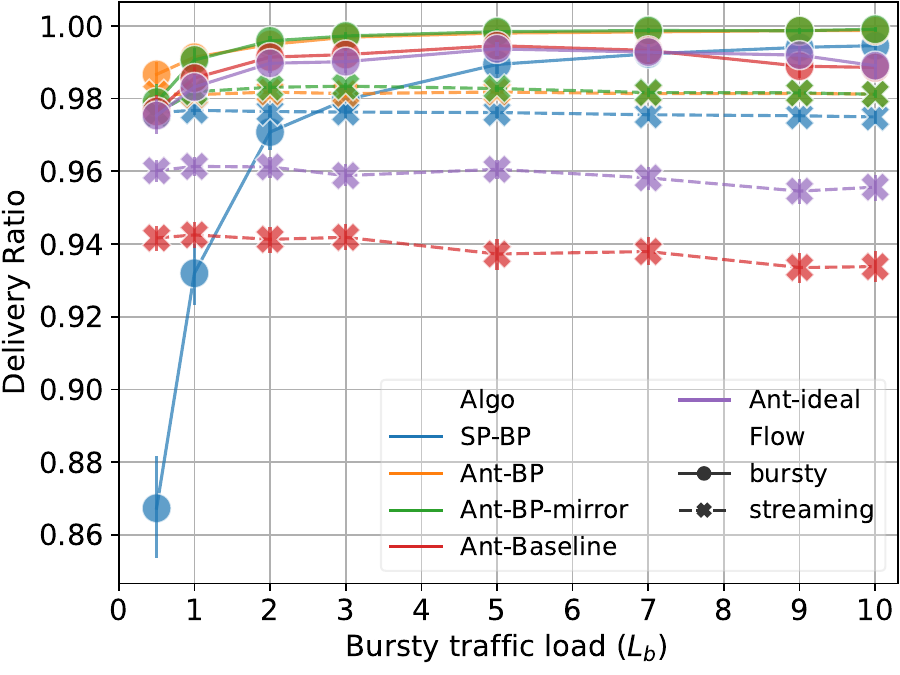}
        \label{fig:mixed:deliveryalllink}
    }
    \hspace{-3mm}
    \subfloat[]{
        \includegraphics[width=0.32\linewidth]{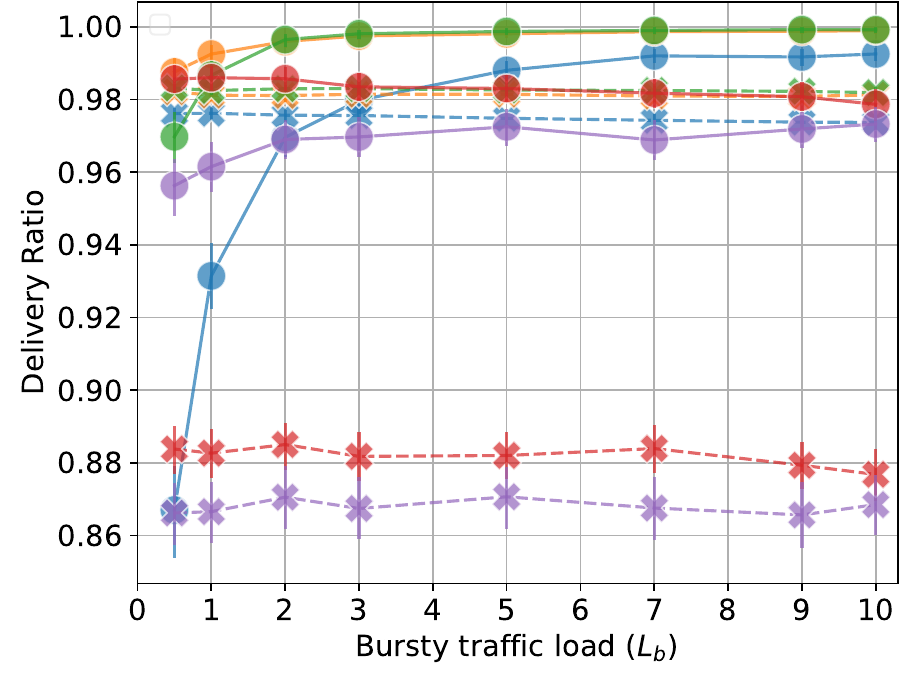}
        \label{fig:mixed:deliverybw}
    }
    \hspace{-3mm}
    \subfloat[]{
        \includegraphics[width=0.32\linewidth]{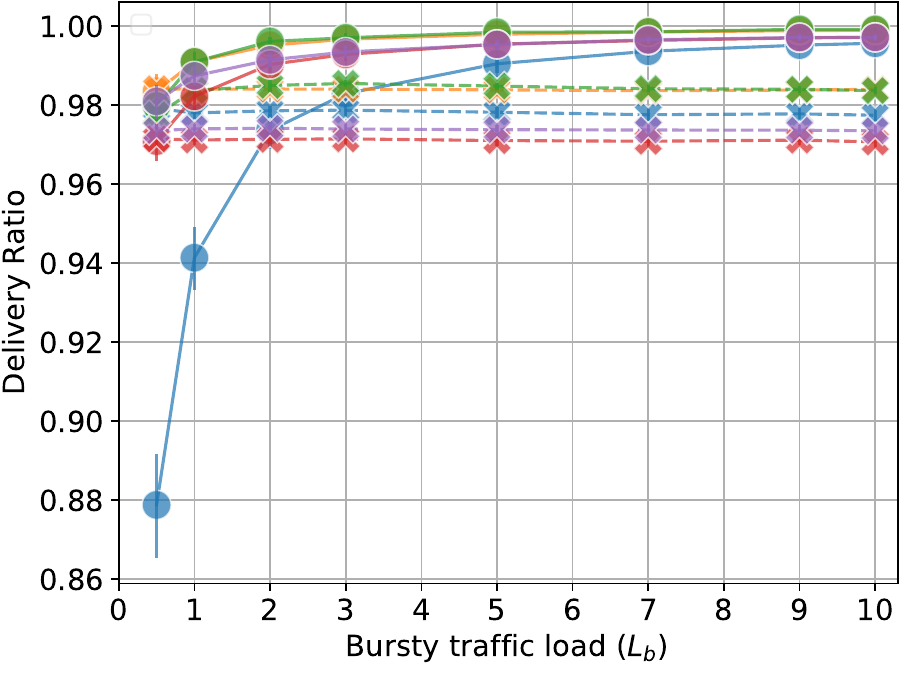}
        \label{fig:mixed:deliverylocal}
    }

    \caption{The average (a, b, c) end-to-end latency and (d, e, f) delivery ratio of tested routing policies as a function of bursty flow load $ L_b $ in 10 random instances of wireless networks of 100 nodes under mixed traffic setting with a constant streaming load  $L_s=1.0$. Subfigures (a, d), (b, e), and (c, f) correspond to the All-Links, BW-Persist, and Local-Persist link failure scenarios, respectively.
    }
    \label{fig:mixedfailure}
    \vspace{-0.1in}
\end{figure*}

\subsection{Resilience under Transient Link Failure}\label{sec:congestion}

We model transient link failure occurring after scheduling in \eqref{eq:policy} but before transmission, evaluating three disruption models. 
In \textit{All-Links}, every link fails independently per time slot with a random probability $p_e\in\mathbb{U}(0,0.05)$ fixed throughout the simulation. 
In \textit{BW-Persist}, we target the top $5\%$ of links by edge betweenness centrality on the connectivity graph $\mathcal{G}^n$.
Failure events on these critical links arrive via a Poisson process with rate $\lambda_e = \nicefrac{p_e}{D_{\text{avg}}}$, and persist for a normally distributed duration centered at $D_{\text{avg}}=20$ time slots.
In such failure events, transmission fails at a fixed probability $p_e\in\mathbb{U}(0,0.05)$.
In \textit{Local-Persist}, this same stochastic failure pattern applies to all links within a randomly selected circular region encompassing $5-6\%$ of network nodes. 
Upon any transmission failure, the affected packets revert to an unscheduled state for that time slot. 
For all ant-based schemes, the pheromone level on the failed link also decays by $5\%$ to discourage their subsequent selections.

To evaluate resilience, we simulate a mixed traffic setting with a fixed streaming load $L_s = 1.0$ and varying bursty loads $L_b \in \{0.5, 1, 2, \dots, 10\}$, and present their average end-to-end latency and delivery ratios as a function of $L_b$  in Figs.~\ref{fig:mixedfailure}.
Across all three failure scenarios, Ant-BP and Ant-BP-mirror consistently experience the least performance degradation, maintaining significantly lower delays and higher delivery ratios than the baseline schemes.

Under independent \textit{All-Links} failures (Figs.~\ref{fig:mixed:delayalllink} and~\ref{fig:mixed:deliveryalllink}), the gradient-based pathfinding of Ant-BP and Ant-BP-mirror yields decreasing delays and increasing delivery ratios for heavier bursty loads ($L_b \geq 7.0$), comfortably outperforming the Ant-Baseline and Ant-Ideal.
This stability advantage becomes highly pronounced under \textit{BW-Persist} failures (Figs.~\ref{fig:mixed:delaybw} and~\ref{fig:mixed:deliverybw}), where failures on critical links severely degrade the performance of Ant-Baseline and Ant-Ideal, underscoring their vulnerability to disruptions on critical paths.

While the ACO baselines appear relatively robust to the randomized \textit{Local-Persist} failures (Figs.~\ref{fig:mixed:delaylocal} and~\ref{fig:mixed:deliverylocal}), the strictly gradient-dependent SP-BP suffers severe starvation at low bursty loads ($L_b \leq 3.0$), causing significant performance drops.
Conversely, at these extremely low burst loads ($L_b \leq 1.0$), Ant-BP achieves higher bursty delivery rates than Ant-BP-mirror across all failure models.
This confirms the structural advantage of mapping lightweight bursty traffic to continuous virtual streaming flows, which stabilizes path formation and significantly mitigates the last-packet problem even in highly disruptive environments.

% In the localized failure scenario shown in Figs.~\ref{fig:mixed:delaylocal} and~\ref{fig:mixed:deliverylocal}, where random network regions (which may or may not include critical links) are affected, Ant-Baseline and Ant-Ideal appear relatively more robust than in the previous cases.
% At low bursty loads ($L_b \leq 2.0\text{--}3.0$)
% SP-BP exhibits higher delays and lower delivery ratios compared with Ant-Baseline and Ant-Ideal. 
% This behavior arises from its gradient-based design, which makes such schemes more susceptible to the last packet problem under light traffic conditions, unlike probabilistic approaches.
% When the bursty load is very low ($L_b \leq 1.0$), Ant-BP achieves higher delivery rates for bursty flows than Ant-BP-mirror across all failure scenarios. 
% This outcome highlights the advantage of treating all traffic as streaming during virtual routing, as in Ant-BP,
% which enables more stable path formation under sparse bursty activity and mitigates the last packet problem.

\begin{figure*}[!t]
    % ---------- Row 1 ----------
    \hspace{-3mm}
    \subfloat[]{
        \includegraphics[width=0.32\linewidth]{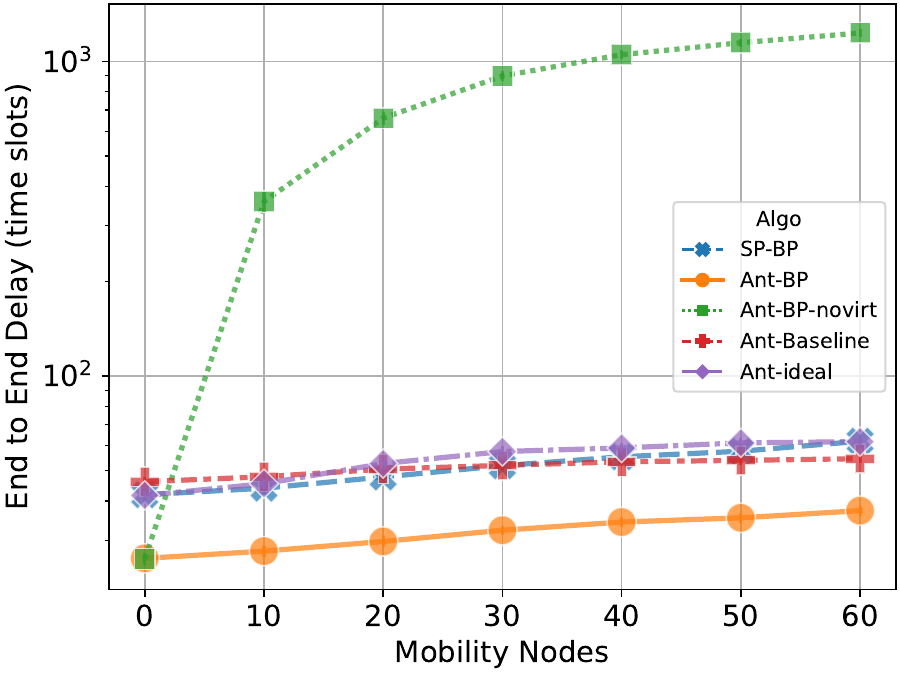}
        \label{fig:stream:Syndelay}
    }
    \hspace{-3mm}
    \subfloat[]{
        \includegraphics[width=0.32\linewidth]{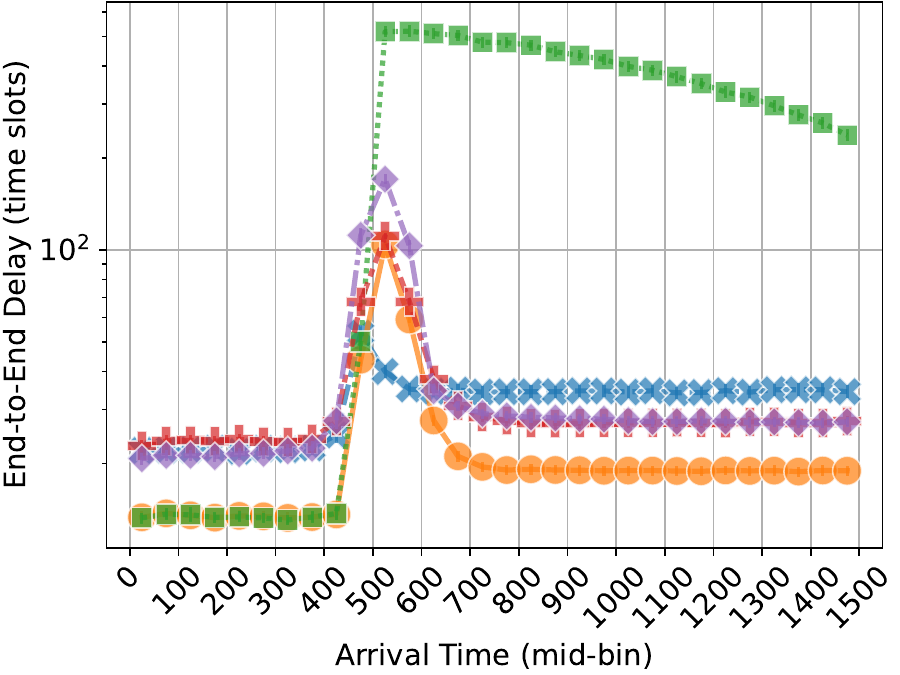}
        \label{fig:stream:delayarrival60}
    }
    \hspace{-3mm}
    \subfloat[]{
        \includegraphics[width=0.32\linewidth]{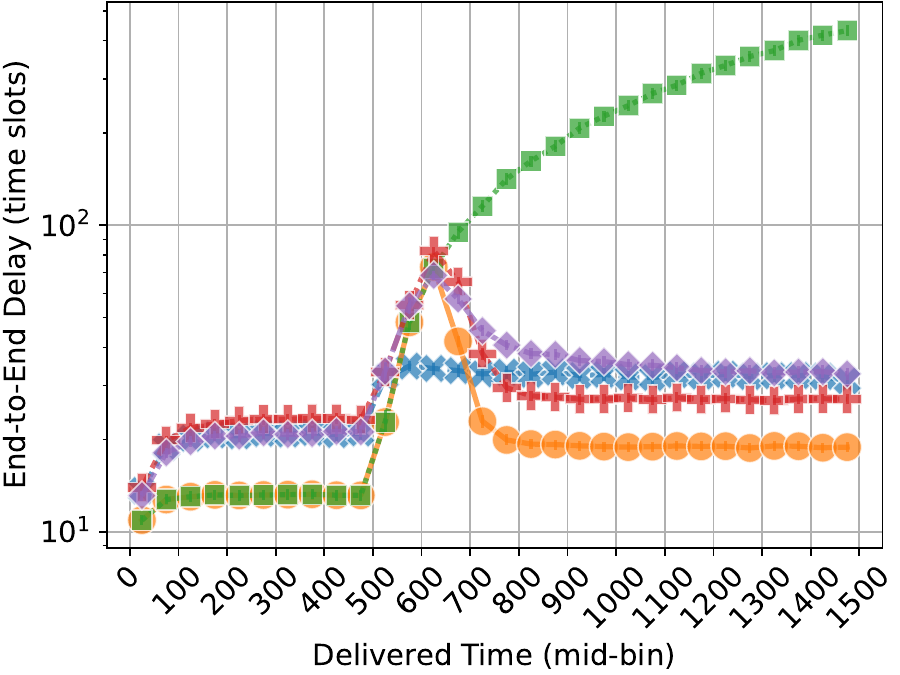}
        \label{fig:stream:delaydelivered60}
    }

    % Small vertical space between rows
    \vspace{1mm}

    % ---------- Row 2 ----------
    % \hspace{-5mm}
    % \subfloat[]{
    %     \includegraphics[width=0.34\linewidth]{Figures/Undelivered_streaming1.0_mobnodes60_mobility_count.pdf}
    %     \label{fig:stream:undelivercountmob60}
    % }
    % \hspace{-4.5mm}
    % \subfloat[]{
    %     \includegraphics[width=0.34\linewidth]{Figures/Delivered_streaming1.0_mobnodes60_mobility_count.pdf}
    %     \label{fig:stream:delivercountmob60}
    % }
    % \hspace{-5.5mm}
    % \subfloat[]{
    %     \includegraphics[width=0.34\linewidth]{Figures/Delivered_streaming1.0_mobnodes60_mobility_delay.pdf}
    %     \label{fig:stream:deliverdelaymob60}
    % }

    \caption{The average (a) end-to-end latency of all packets as a function of mobility level (number of mobile nodes), (b) end-to-end latency of delivered packets as a function of packet arrival time (binned in ranges of 50), and (c) end-to-end latency of delivered packets as a function of packet delivery time (binned in ranges of 50) for tested routing policies in 10 random instances of wireless networks of 100 nodes under streaming traffic setting with a constant streaming load $L_s=0.5$. 
    }
    \label{fig:mobility}
    \vspace{-0.1in}
\end{figure*}

\subsection{Adaptation under Note Mobility}\label{sec:results:mobility}

\begin{table}[h!]
\setlength{\tabcolsep}{5pt}
\centering
\caption{Link removal ratios across mobility levels.}
\label{tab:mobility_ratios}
\footnotesize
\begin{tabular}{lcccccc}
\toprule
Mobile Nodes & $10$ & $20$ & $30$ & $40$ & $50$ & $60$ \\
\midrule
Link Removal & $16.6\%$ & $33.3\%$ & $48.5\%$ & $61.0\%$ & $71.6\%$ & $80.4\%$ \\
\bottomrule
\end{tabular}
\end{table}

We simulate node mobility using a constrained Gaussian random walk model, relocating a random subset of nodes (e.g., $\{10, 20, 30, 40, 50, 60\}$) during physical routing with 2D steps  drawn from $\mathbb{N}(\mathbf{0}, 0.1^2\mathbf{I})$.
Crucially, these perturbations are bounded within the 2D square and restricted to preserve the connectedness of the network (e.g., $\ccalG^n$), while causing some existing links break and new ones form at high probabilities.
The link removal ratios by mobility levels are listed in Table~\ref{tab:mobility_ratios}.
Following the mobility adaptation mechanisms in Sec.~\ref{sec:adaptation}, stranded packets are remapped as virtual sources during the subsequent virtual routing phase to adaptively redistribute traffic and backlogs across the updated topology.
In our simulation horizon ($T=2000$), mobility triggers at $t=500$, prompting the scheduled virtual routing phase at $t=600$.
To account for signaling overhead, physical routing is paused for $T'=10$ time slots while virtual routing executes for $K=1000$ steps.
During this pause, newly arriving exogenous packets are queued at their source nodes, and physical routing resumes at $t=610$ with the updated pheromones. 
The choice of $T'$ depends on the actual MaxWeight scheduling overhead.
Although the evaluation only covers node movements in a single update cycle, this periodical process naturally extends to continuous long-term horizons and dynamic node churn (i.e., joining/leaving networks).
In addition, pure streaming traffic is used and the number of flows is drawn uniformly between $15$ and $30$. 

We compare Ant-BP against SP-BP, Ant-Baseline, Ant-Ideal, and an ablation variant, \textit{Ant-BP-novirt}, which lacks periodic virtual updates (new links instead initialize with the mean pheromone value, and $T'=0$).
Fig.~\ref{fig:stream:Syndelay} presents the end-to-end delay in a logarithmic scale as a function of mobile nodes.
The streaming load is $L_s=0.5$, and the latency of undelivered packets is set as their residency time ($T-$ injection time).
As expected, delay increases with the mobility degree, which induces link removal ratios ranging from $16.6\%$ ($10$ mobile nodes) to $80.4\%$ ($60$ mobile nodes).
Ant-BP achieves the best performance across all mobility levels, whereas Ant-BP-novirt severely degrades even under mild mobility (delay of $358.5$ time slots at $10$ mobile nodes), isolating periodic virtual routing as the key mechanism for topology adaptation.
While SP-BP and Ant-Ideal handle mild mobility well, their strict gradient reliance and path-dependent updates, respectively, cause them to falter at higher mobility degrees ($>20$ and $>40$ nodes), where Ant-Baseline slightly outperforms.
In the extreme $60$-node scenario, Ant-BP establishes the lowest delay ($37.2$), decisively outperforming Ant-Baseline ($54.5$), Ant-Ideal ($61.7$), SP-BP ($61.9$), and Ant-BP-novirt ($1231.6$).

Unlike the immediate per-link updates of Ant-Baseline, Ant-Ideal relies on end-to-end path-cost evaluations via proactive ants, rendering its pheromone distributions highly sensitive to mobility-induced link breakages.
Furthermore, Ant-Baseline's superior resilience compared to SP-BP at higher mobility degrees underscores the inherent advantage of probabilistic multi-path routing over strictly gradient-dependent methods in highly dynamic environments.

We further analyze the temporal evolution of delivered packets in the extreme $60$-node scenario, by presenting the end-to-end delay as a function of packet injection time (Fig.~\ref{fig:stream:delayarrival60}) and delivery time (Fig.~\ref{fig:stream:delaydelivered60}).
In these plots, the time horizon is divided into discrete bins of $50$ time slots (represented by their mid-bin values) to capture the moving average of routing delay across the adaptation timeline.

In Fig.~\ref{fig:stream:delayarrival60}, all schemes exhibit a sharp delay peak for packets injected immediately prior to the $t=500$ topology shift.
Following this event, the delay of Ant-BP-novirt remains permanently elevated, demonstrating the catastrophic failure of static probabilistic fields under topological changes.
Conversely, the slot-by-slot dynamic SP-BP recovers quickly, showing a drop in delay for packets injected shortly after the shift, but ultimately plateaus at a persistently higher delay ($\approx 34.4$) post-adaptation.
The periodic schemes (Ant-BP, Ant-Baseline, and Ant-Ideal) experience escalating delays for packets arriving during the $t \in [500, 600]$ window, but exhibit dramatic recovery for packets injected after the $t=610$ virtual update.
Ultimately, Ant-BP achieves the absolute lowest delay for packets arriving post-adaptation ($\approx 19.0$), comfortably outperforming Ant-Baseline ($\approx 27.2$) and Ant-Ideal ($\approx 27.4$).

From the perspective of delivery time (Fig.~\ref{fig:stream:delaydelivered60}), the delay surge for the periodic ant-based methods appears broader, spanning the $525 \le \text{mid-bin} \le 625$ window.
This prolonged hump reflects the gradual flushing of stranded packets that lingered in the network while awaiting the $t=610$ topology update, whereas the delivery delays of SP-BP remain entirely steady after its initial spike.
However, once these transient backlogs clear post-adaptation, Ant-BP conclusively yields the lowest delivery delay, followed by Ant-Baseline, SP-BP, and Ant-Ideal.
Collectively, these temporal dynamics reveal a fundamental routing trade-off: while strict gradient-based routing (SP-BP) provides immediate reactivity to mobility, it often settles into less efficient long-term paths.
By absorbing temporary congestion while awaiting a structured virtual update, Ant-BP establishes superior, globally optimal routing probabilities that definitively minimize long-term latency in dynamic networks.

\section{Conclusions and Future Work}

In this paper, we introduce Ant Backpressure (Ant-BP) routing, a distributed framework bridging the theoretical optimality and efficiency of gradient-based routing with the practical constraints of legacy forwarding architectures.
By enabling efficient link-capacity sharing and mapping short-lived bursty traffic into virtual streaming flows, Ant-BP structurally eliminates the last-packet problem inherent to SP-BP.
Furthermore, Ant-BP demonstrates superior resilience in highly dynamic mobile environments. It gracefully mitigates the disruptions from transient link failures that typically plague reactive protocols, and under extreme mobility scenarios, its periodic virtual routing mechanism trades temporary congestion for globally superior probabilistic paths. This yields significantly lower long-term latency than the myopic, slot-by-slot reactions of SP-BP.
The primary trade-off for this architectural simplicity is a slight goodput reduction under heavy traffic compared to SP-BP.

Future work will explore proactive route maintenance to enhance scalability and address continuous dynamic node churn in realistic deployments.
Additionally, our approach to approximating constrained multi-commodity min-cost flow (MCMCF) problems sheds light on resolving non-linearity induced by resource contention, offering insights not only for wireless routing but also for broader domains such as edge computing.

% In this study, we propose to improve the performance of short-lived traffic in dynamic wireless multi-hop networks, particularly under the interference of streaming traffic and in the presence of link failures and node mobility.
% To achieve this, we integrate the superior path-finding capability of SP-BP, which maximizes queue stability, into ACO routing with FIFO link scheduling.
% This approach enables packets of different commodities to share link capacity efficiently and mitigates the last-packet problem for short-lived flows.
% Our proposed scheme is shown to be effective for mixed traffic and robust to mismatched virtual flow configurations, while maintaining robustness across a range of link failure scenarios.
% It also effectively adapts to high degrees of node mobility through its periodic virtual route updates and can achieve throughput similar to SP-BP under low-to-medium traffic intensity.
% However, the cost of these benefits is lower throughput under heavier traffic. 
% Future work includes developing continuous route maintenance mechanisms to improve throughput under heavier traffic and mobility scenarios characterized by the addition and removal of nodes, alongside changes in link connectivities, to better reflect realistic ad-hoc wireless dynamic networks.

\appendix
\section{Ant Colony Optimization}
% \subsection{ACO baseline}
\label{app:aco}
The baseline ACO heuristic updates routing policy as follows.  
At time step $t$, a data packet that has arrived at a non-destination node $i \in \mathcal{V}$ is assigned to one of its neighbors $j\in\ccalN_{\ccalG^n}(i)$ with probability $p_{ij}^{(c)}(t)$ given by~\cite{zhang2017survey}
\begin{equation}
    \label{eq:sgs}
    p_{ij}^{(c)}(t) = \frac{\big[\rho_{ij}^{(c)}(t)\big]^\alpha \cdot h_{ij}^\beta}{\sum_{l\in\ccalN_{\ccalG^n}(i)} \big[\rho^{(c)}_{il}(t)\big]^{\alpha} \cdot h_{il}^{\beta}},
\end{equation}
where $\rho^{(c)}_{ij}(t)$ is the pheromone intensity associated with commodity $c$ on link $({ i, j})$ at time step $t$, $h_{ij}$ is a heuristic cost of link $({ i, j})$, and $\alpha$ and $\beta$ are parameters for tuning the importance of the pheromone intensity and cost. 
%heuristic information, respectively. 
Assuming that $m(c,t)$ ants reach destination $c$ at time step $t$, the pheromone intensity is updated as~\cite{zhang2017survey}
\begin{subequations}
\begin{align}
    \rho^{(c)}_{ij}(t+1) &= (1 - \varepsilon) \cdot \rho^{(c)}_{ij}(t) + \sum_{k=1}^{m(c,t)} \theta_{ij,k}^{(c)}(t),\label{eq:ph_update}\\
    \theta_{ij,k}^{(c)}(t) &= \begin{cases} 
    \left[\phi\big(\ccalP^{(c)}_k\big)\right]^{-1}, & \text{if } ({ i, j})\in \ccalP^{(c)}_k \\ 
    0, & \text{if } ({ i, j})\notin \ccalP^{(c)}_k  
    \end{cases},\label{eq:ph_update:delta}
\end{align}    
\end{subequations}
where $\varepsilon$ is the evaporation rate, $\theta_{ij,k}^{(c)}(t)$ is the amount of pheromone deposited by the $k$th ant, $\ccalP^{(c)}_k $ is the path of the $k$th ant, and $\phi\big(\ccalP^{(c)}_k\big)$ is the cost function of path $\ccalP^{(c)}_k $.

\printcredits

\section*{Acknowledgments}
Research was sponsored by the U.S. Army Combat Capabilities Development Command (DEVCOM) Army Research Office and was accomplished under Cooperative Agreement Number W911NF-24-2-0008 and W911NF-19-2-0269. The views and conclusions contained in this document are those of the authors and should not be interpreted as representing the official policies, either expressed or implied, of the Army Research Office or the U.S. Government. The U.S. Government is authorized to reproduce and distribute reprints for Government purposes notwithstanding any copyright notation herein.
% Elsevier-specific back matter, typically acknowledgements / declarations.
% \input{shared/backmatter_elsevier.tex}
\bibliography{references}

\end{document}